\documentclass[a4paper,11pt]{article}

\usepackage[
    a4paper,
    left=1.8cm,
    right=1.8cm,
    top=2.0cm,
    bottom=2.0cm
]{geometry}

\usepackage[numbers,sort&compress]{natbib}

\usepackage{amssymb,amsmath,amsthm}
\usepackage{graphicx}
\usepackage{subcaption}
\usepackage{pgf,tikz}
\usetikzlibrary{arrows,arrows.meta,positioning}
\usepackage{hyperref}
\usepackage{float}
\usepackage{placeins}
\usepackage{float}

\newtheorem{proposition}{Proposition}
\newtheorem{Lemma}{Lemma}

\def\tsc#1{\csdef{#1}{\textsc{\lowercase{#1}}\xspace}}
\tsc{WGM}
\tsc{QE}
\tsc{EP}
\tsc{PMS}
\tsc{BEC}
\tsc{DE}

\title{Integrable and Chaotic 4-Dimensional Lotka--Volterra Models and Population Sustainability} \author{ H. Christodoulidi\thanks{Corresponding author: \texttt{hchristodoulidi@lincoln.ac.uk}} \\ College of Health and Science, University of Lincoln \and T.E. Kouloukas \\ School of Computing and Digital Media, London Metropolitan University \and L. B. Drossos \\ Department of Electrical and Computer Engineering, University of Peloponnese \and T. Bountis \\ Department of Mathematics, University of Patras, Patras} \date{} 
\begin{document} \maketitle


-----------------------------------------------------------


\begin{abstract}
We examine a 4-dimensional generalization of predator-prey models describing interactions among four species. We show that these systems exhibit a rich variety of asymptotic behaviors, with solutions that are either integrable, leading to non-sustainable dynamics characterized by unbounded population growth or collapse, or chaotic, where all species coexist over time. Moreover, we review the integrable cases, which are classified into two-parametric families, and examine three-parametric families that behave similarly to the known two-parametric integrable cases. We also analyze chaotic Lotka--Volterra models in terms of their Lyapunov exponents and  Poincar{\'e} surfaces of section, which turn out to reveal underlying structures that are compatible with what is expected from sustainable dynamics.   \end{abstract}


\noindent\textbf{Keywords:} Lotka--Volterra; Integrability; Chaos; Sustainability

\maketitle

\section{Introduction}
In this paper we study a 4-dimensional parametric class of predator-prey systems in terms of their dynamics and asymptotic behavior. This type of models are well-known for their significant applications in ecology, epidemics and finance. Particularly, in ecosystems such models describe the population dynamics between interacting species acting as predator, prey or both. 

One of the most famous predator-prey models is the well-known Lotka--Volterra 2-dimensional system, originally developed
independently by Alfred J. Lotka \cite{lotka} in 1925 and Vito Volterra \cite{volterra} in 1926, to describe the evolution of populations of two species, one predator and one prey. This system has one stable fixed point around which both populations co-exist at all times in closed orbits. Such sustainable solutions arise under the assumption that, in the absence of interaction, the prey population $x$ would grow exponentially, whereas the predator population $y$ would rapidly decline toward extinction. Allowing for quadratic interactions between the two species, balances  prey's exponential growth and  predator's decline, transforming their dynamics into a sustainable ecosystem with periodic behavior. 

From a dynamical systems perspective, this behavior follows from the existence of an elliptic equilibrium point in the first quadrant of the phase space, where bounded motion explicitly results from a conserved quantity $I = \delta x + \beta y - \gamma \ln{x} - \alpha \ln{y}$, $\alpha, \beta, \gamma, \delta >0$. In economics, a variation of this model is known as the Goodwin model \cite{goodwin}, and describes the cyclical interaction between the wage share of national income and the employment rate.  

There are several works published on generalizations of the classical Lotka--Volterra system \cite{balle, bog, bountis, chara, HK2019, DEKV, hern, itoh09,  KQV,RZ2025, RagniscoZullo2025, ScaliaRagniscoTirozziZullo2024, KKQTV, VKMcQR2025} in higher dimensions that consider various interaction patterns between populations and different assumptions on their reproductive dynamics. Numerical studies of four-dimensional  Lotka--Volterra systems displaying chaos can be found e.g. in \cite{HK2019, vano}.
A fundamental question which thus arises is which of these models are integrable and which display chaotic dynamics. Over the last decades, several authors have studied the integrability properties of such predator-prey models and published numerous research papers on selections of the system's parameters under which the system is Liouville integrable, thus providing a complete set of conserved quantities that are functionally independent and in involution \cite{balle, bog, chara, HK2019, DEKV, hern, itoh09, KQV, KKQTV}.

In this work, we investigate the behavior of a class of inhomogeneous Hamiltonian Lotka--Volterra systems, whose integrability for various cases of the parameters has been studied in \cite{HK2019,VKMcQR2021,VKMcQR2025} using Poisson--Lie methods. 
These Lotka--Volterra systems represent the competition between a hierarchical $N$-species interaction, where the highest-ranked species acts as a apex-predator, while the rest have a mixed predator-prey role, down to the lowest-ranked species, which serve as prey to all the others. For this class of systems, an open question arises whether the known classification is exhaustive or whether additional integrable cases exist for different choices of the system's parameters. In this paper, we focus on the integrability and chaotic behavior of $N=4$ interacting species. Alongside, we also investigate the system's parameters that lead to sustainable solutions, where all four populations co-exist and remain balanced at all times. The paper is organized as follows:

In Section \ref{themodel} we describe the model's dynamics and the log-canonical Poisson structure which yields the Hamiltonian function. In Section \ref{stabanal} we perform a classical stability analysis on the system's equilibrium points and provide appropriate parameter conditions (Proposition 1) for the existence of elliptic equilibrium points, which guarantee sustainability through bounded motion for all four populations. 

In subsection \ref{integrablecases} we review the existing integrable cases for this type of predator-prey models, and provide explicit expressions 
for the functionally independent and commuting integrals. It turns out that there are three fundamental integrals generating six two-parametric families of integrable 4D Lotka--Volterra models. 

In subsection \ref{chaoscases} we study chaotic cases of a 4-dimensional Hamiltonian Lotka--Volterra system by providing numerical results on the Lyapunov exponents and Poincar\'e surfaces of sections, which reveal rich dynamical behavior and complexity. Most importantly, chaotic cases are identified as the only ones that yield bounded solutions associated with the species coexistence, leading to the interesting conclusion that chaos is vital for long-term survival in these models.
Finally, in  subsection \ref{newcases} we extend our study to three-parametric families of 4D models to examine whether these more general systems fall into either class of integrable or chaotic cases. 

\section{The model \label{themodel}}
We consider a parametric class of Hamiltonian  Lotka--Volterra systems of the form \cite{HK2019}
\begin{equation}\label{system1}
  \dot x_i=x_i\left(\sum_{j>i}^n a_jx_{j}-\sum_{j<i}^n a_jx_{j}+r_i\right)\;, \ \ \ \ a_i,r_i \in \mathbb{R}.
\end{equation}
In this paper, we focus on the case of $n=4$ and $a_i>0$, i.e. to the four-dimensional system 
\begin{equation}\label{system2}
 \mathbf{\dot x} = \mathbf{x} \odot  (A \mathbf{x} + \mathbf{r})  
\end{equation}
where here $\mathbf{x}=(x_1,x_2,x_3,x_4)$, $\mathbf{r}=(r_1,r_2,r_3,r_4)$ and 
\begin{equation}\label{Amatrix}
  A=\begin{pmatrix}
  0 & a_2& a_3& a_4 \\
  -a_1 & 0& a_3& a_4 \\
  -a_1 & -a_2& 0& a_4 \\
   -a_1 &-a_2&-a_3& 0
\end{pmatrix}\;, \ \text{with} \  a_i>0,  
\end{equation}
is the community matrix. The operation $\odot$ denotes the Hadamard element-wise product between the two vectors and is not merely a notational convenience. Rather, it enforces the invariance of the extinction hyperplanes ($x_i=0$), thus capturing mathematically the biological reality that a species driven to zero cannot spontaneously recover. 

\begin{figure}
\centering
\begin{tikzpicture}[
    node distance=1.8cm,
    every node/.style={circle, draw, minimum size=0.2cm},
    >=Stealth
]

\node (x1) {$x_1$};
\node (x2) [below left of=x1] {$x_2$};
\node (x3) [below right of=x1] {$x_3$};
\node (x4) [below=1.8cm of x1] {$x_4$};

\draw[->] (x2) -- (x1);
\draw[->] (x3) -- (x1);
\draw[->] (x3) -- (x2);
\draw[->] (x4) -- (x2);
\draw[->] (x4) -- (x3);
\draw[->] (x4) -- (x1);
\end{tikzpicture}
\caption{Schematic representation of the hierarchical `four species' interactions: the couplings between the populations of apex-predator $x_1$, prey of all $x_4$, and the mixed predator-prey species $x_2$ and $x_3$. The arrows denote the dynamical flow. \label{graph}}
\end{figure}
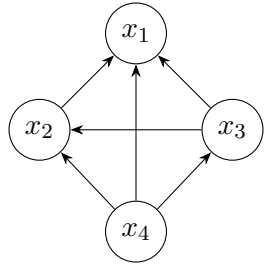

\begin{figure}
\centering
    \includegraphics[width=0.8\linewidth]{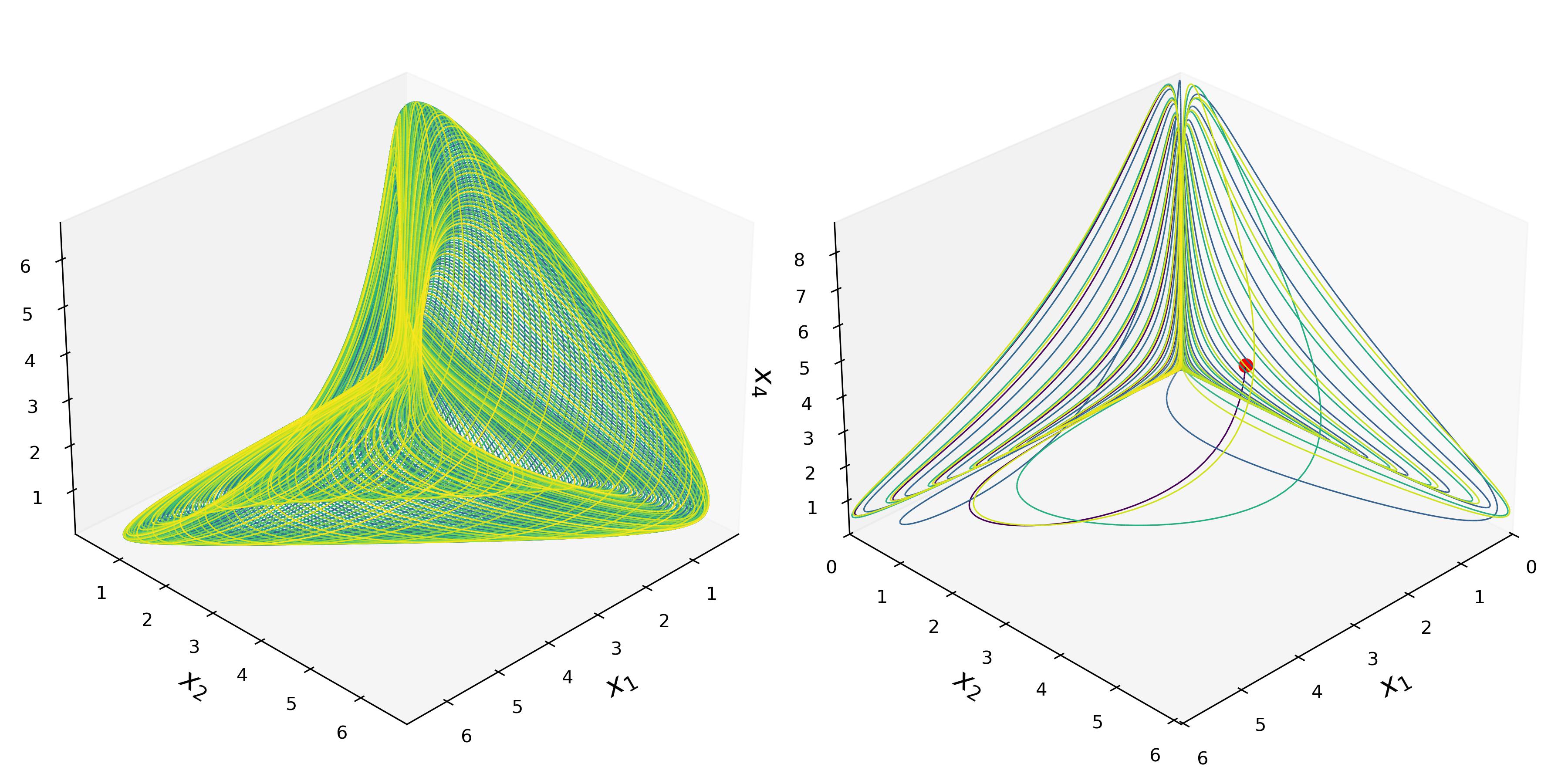}
        \includegraphics[width=0.8\linewidth]{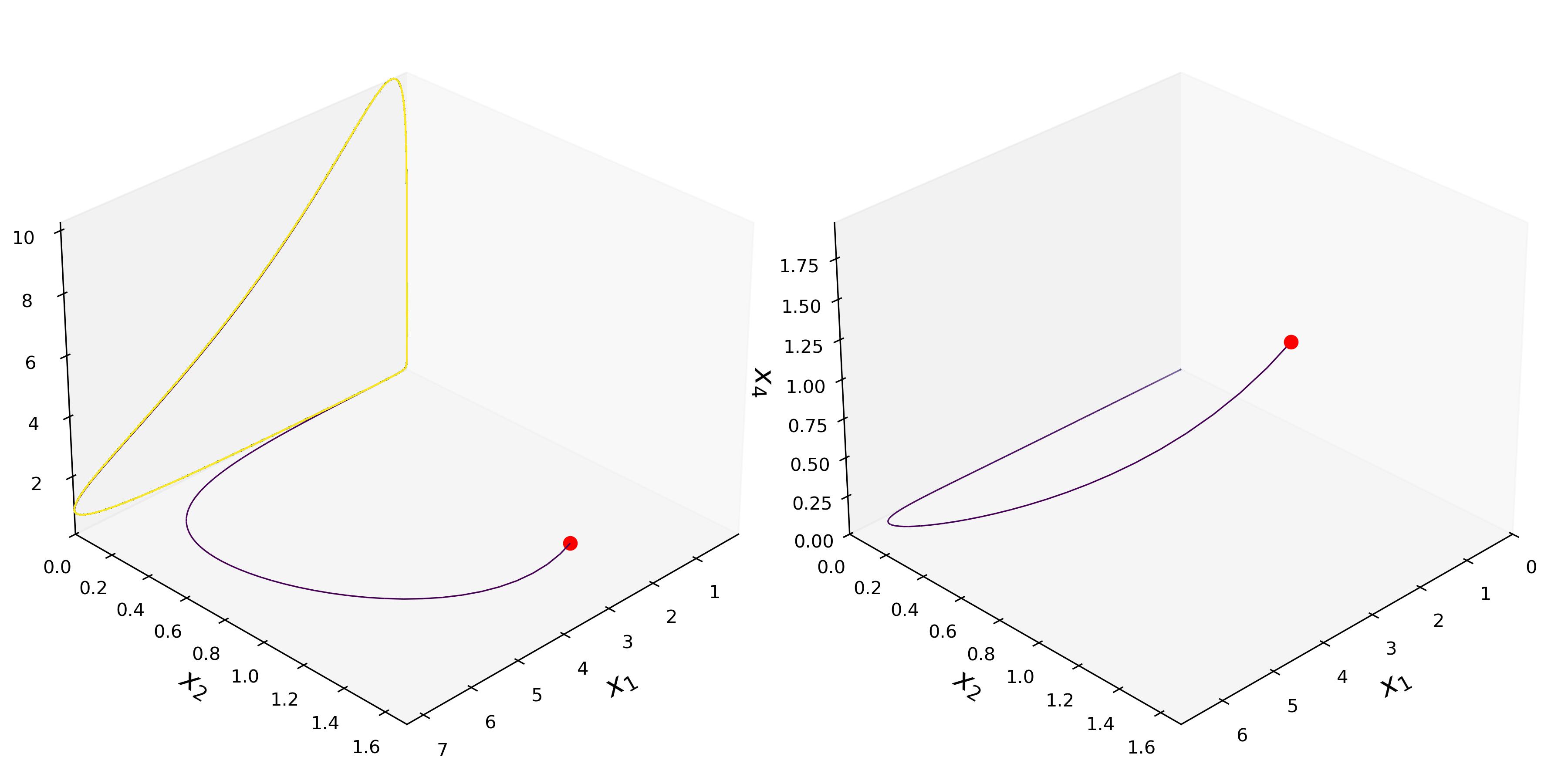}
        \caption{ Three-dimensional projections on $x_1,x_2,x_4$ of a trajectory of the system (\ref{system1}) for $a_i =1$, $i=1,\ldots,4$ at energy level $E=8$.  
        Top: (a) $k_1=k_2 =k_3 = k_4=-1$, (b) $k_1=k_2 = - 0.05, k_3 = k_4=-1$. Bottom: (c) $k_1=k_2 =0, k_3 = k_4=-1$, (d) $k_1= k_2 = 1, k_3 =  k_4=-1$. \label{fig1k}}
\end{figure}

The graph in Fig.\ref{graph} illustrates the hierarchy in the roles of the four species populations in terms of their interaction network. The top variable $x_1$, represents the population of the predator for the other three species populations $x_2,x_3$ and $x_4$. Similarly, $x_2$ is the predator of the two populations $x_3,x_4$, while $x_3$ preys exclusively on $x_4$. The value of the coefficients $a_i$ in front of the non-linear terms determine the strength of the interaction.

The system (\ref{system2}) is Hamiltonian with respect to the log-canonical Poisson structure 
\begin{equation}\label{poisson}
\{x_i,x_j\}=x_ix_j\;,\quad 1\leq i<j \leq 4 
\end{equation} 
and the Hamiltonian 
\begin{equation} \label{hamilton}
H(\mathbf{x})=\sum_{i=1}^4 (a_i x_i+k_i\log x_i) , \quad \text{when} \quad x_i>0 ,
\end{equation}
where $\mathbf{k}=(k_1,\ldots, k_4)$ are the constants 
\begin{equation} \label{k1to4}
\begin{aligned}
    k_1 &=-r_2 + r_3 - r_4  \\
    k_2&=r_1 - r_3 + r_4  \\
    k_3&=- r_1 + r_2 - r_4 \\
    k_4&=r_1 - r_2 + r_3 \quad.
\end{aligned}
\end{equation}

Furthermore, for our case where $a_i>0$, all the $a_i$ parameters  can be rescaled to $1$ without loss of generality \cite{KQV}. 
Consequently, the system takes the equivalent form 
\begin{equation} \label{oursystem}
\begin{aligned}
    \dot x_1 &= x_1 \left( x_2 + x_3 + x_4 + r_1  \right)   \\
    \dot x_2 &= x_2 \left( -x_1 + x_3 + x_4 + r_2  \right)   \\
    \dot x_3 &= x_3 \left( -x_1 - x_2 + x_4 + r_3  \right)   \\
    \dot x_4 &= x_4 \left( -x_1 - x_2 - x_3 + r_4  \right)   
\end{aligned}
\end{equation} with Hamiltonian 
\begin{equation} \label{hamilton2}
H(\mathbf{x})=\sum_{i=1}^4 (x_i+k_i\log x_i) , \quad \text{for} \quad x_i>0 .
\end{equation}

Our aim here is to study how the dynamical behavior of Hamiltonian system \eqref{oursystem} depends on the choice of the various of the parameters $k_i$. 

In Fig.\ref{fig1k} we display 3D projections of 
four representative numerical solutions, which result to distinct dynamical behavior. In these four examples we keep $k_3=k_4=-1$ and vary $k_1, k_2$. In Fig.\ref{fig1k}(a), where all $k_i$ are negative, the motion resembles winding around a `torus', 
yielding sustainable solutions for all populations.  The `butterfly' trajectory in Fig.\ref{fig1k}(b), corresponding to the near-boundary case of $k_1, k_2= -0.05$, transitions between $(x_1,x_4)$ and 
$(x_2,x_4)$ planes. This behavior can be interpreted as a competition between the two predators $x_1$ and $x_2$, where for certain periods $x_1$ preys predominantly on $x_4$ with $x_2$ being negligible and in other periods $x_2$ dominates over $x_1$. While case (b) is also bounded, it is close to an extinction threshold for the two predators under external perturbations.  
The red marker in each panel denotes the initial condition, while the trajectory coloring from dark blue to yellow represents the arrow of time. Figs.\ref{fig1k}(c) and (d) represent non-sustainable solutions. When two of the $k_i$, values are zero, for instance  when $k_1=k_2 =0$ as it is in Fig.\ref{fig1k}(c), the trajectory approaches the $(x_1,x_4)$ plane exponentially. Instead, when we consider two positive and two negative $k_i$ values, as in Fig.\ref{fig1k}(d) where $k_1= k_2 = 1, k_3 =  k_4=-1$, the trajectory collapses at the origin and all populations go extinct. 
These four examples illustrate the diversity in system's asymptotic behavior, which is explained in the next section through the stability analysis of the equilibrium points.

\section{Linear stability analysis \label{stabanal}}

A simple linear stability analysis  reveals a variety of dynamical behaviors exhibited by the system.
In particular, the eigenvalues of the Jacobian determine the local asymptotic behavior of the trajectories. In a generic case, where both linear coefficients $r_i$ and their transformed values $k_i$, $i=1,\ldots,4$ are non-zero, the above system (\ref{oursystem}) contains 8 equilibrium points. When either $k_i$ or $r_i$ take zero values, the number equilibrium points reduces, as some points fall into the same case.

Table \ref{table} classifies all equilibrium points and their eigenvalues, which determine the stability of nearby orbits. The following proposition focuses on the  cases yielding an elliptic equilibrium point, for which the evolution of every population $x_i(t)$ remains locally bounded over time.

\begin{table}[h]
\centering
\begin{tabular}{|c|c|c|}
\hline
\textbf{Location} & \textbf{Equilibrium Point} & \textbf{Eigenvalues} \\
\hline
Origin & $(0,0,0,0)$ & $\lambda_{1,2,3,4} = r_{1,2,3,4}$ \\
$x_1=x_2=0$ plane & $(0,0,-r_3,r_4)$ & 
$\lambda_{1} = k_2$, $\lambda_{2} = - k_1$,
$\lambda_{3,4} = r_3 + r_4 \pm \sqrt{r_3 r_4} $ \\
$x_1=x_3=0$ plane & $(0,r_4,0,-r_2)$ & 
$\lambda_{1} = - k_3$, $\lambda_{2} =  k_1$,
$\lambda_{3,4} = \pm \sqrt{r_2 r_4} $ \\
$x_1=x_4=0$ plane & $(0,r_3,-r_2,0)$ & $\lambda_{1} = k_{4}, \lambda_{2} = -k_{1} $,
$\lambda_{3,4} = \pm \sqrt{r_2 r_3} $ \\
$x_2=x_3=0$ plane & $(r_4,0,0,-r_1)$ & $\lambda_{1} = k_{3}, \lambda_{2} = - k_2  $,
$\lambda_{3,4} = \pm \sqrt{r_1 r_4} $ \\
$x_2=x_4=0$ plane & $(r_3,0,-r_1,0)$ & $\lambda_{1} = k_{4}, \lambda_{2} = - k_2$, $\lambda_{3,4} = \pm \sqrt{r_1 r_3} $ \\
$x_3=x_4=0$ plane & $(r_2,-r_1,0,0)$ & $\lambda_{1} = k_{4}, \lambda_{2} = -k_3$, $\lambda_{3,4} = \pm \sqrt{r_1 r_2} $ \\
Interior (if $k_i \ne 0$) & $(-k_1,-k_2,-k_3,-k_4)$ & $\lambda_{1,2,3,4}=\pm i 
\sqrt{\alpha \pm \sqrt{\alpha^2- 4 \beta}}/\sqrt{2}$ \\
\hline
\end{tabular}
\caption{The 8 possible equilibrium points of system (\ref{oursystem}) and their eigenvalues. The first equilibrium point is the origin and has all eigenvalues real. The following 6 points have at least two zero components and at least two real eigenvalues $\lambda_1, \lambda_2$. The last equilibrium point (interior) is elliptic for appropriate signs of $\alpha, \beta$, as defined in (\ref{alphabeta}).
}
\label{table}
\end{table}

\begin{proposition}
The 4-dimensional Lotka--Volterra system (\ref{oursystem})
 possesses a single elliptic point at $$\mathbf{x}_0=-\mathbf{k}=(r_2 - r_3 + r_4, -r_1 + r_3 - r_4, r_1-r_2+r_4, -r_1 + r_2 -r_3)$$ 
if and only if  
\[
\sum_{i<j}  k_i k_j >0  \quad \text{and} \quad \prod_{i=1}^4  k_i >0 \quad.
\]
\end{proposition}

\begin{proof}
The system $ \mathbf{\dot x} = f(\mathbf{x})= \mathbf{x} \odot  (A \mathbf{x} + \mathbf{r})  $ with $a_i=1$
under the Hadamard notation directly yields the equilibrium point  $ \mathbf{x}_0 = - A^{-1} \mathbf{r}=-\mathbf{k}$. The Jacobian matrix of the system, 
\[
Df = 
\begin{pmatrix}
x_2 + x_3 + x_4 + r_1 & x_1 & x_1 & x_1 \\
-x_2 & -x_1 + x_3 + x_4 + r_2 & x_2 & x_2 \\
-x_3 & -x_3 & -x_1 - x_2 + x_4 + r_3  & x_3 \\
-x_4 & - x_4 & - x_4 & -x_1 - x_2 - x_3 + r_4
\end{pmatrix} \quad,
\]
evaluated at $\mathbf{x}_0=-\mathbf{k}$, becomes
\[
Df|_{x_0} = 
\begin{pmatrix}
0 & -k_1 & -k_1 & -k_1 \\
k_2 & 0 & -k_2 & -k_2 \\
k_3 & k_3 & 0 & -k_3 \\
k_4 & k_4 & k_4 & 0
\end{pmatrix} \quad .
\]
Its characteristic polynomial takes the form   
\begin{equation} \label{charact}
\det(Df|_{x_0}  - \lambda I )
=
\lambda^4
+
\alpha \lambda^2
+ \beta
\end{equation}
with coefficients 
\begin{equation} \label{alphabeta}
\alpha = k_1 k_2 + k_1 k_3 + k_1 k_4 + k_2 k_3 + k_2 k_4 + k_3 k_4, \quad \beta = k_1 k_2 k_3 k_4 \quad .
\end{equation} 
The discriminant $\Delta = \alpha^2 - 4 \beta $ of $(\ref{charact})$ is non-negative, hence all squares of the eigenvalues $\lambda^2_i$, $i=1,\ldots,4$ of the Jacobian $Df|_{x_0}$ at $ \mathbf{x}_0 = - \mathbf{k}$ are real and given by
\begin{equation} \label{eig2}
\lambda_{1,2} ^2  = \frac{- \alpha + \sqrt{\alpha^2 - 4 \beta}}{2} , \quad 
\lambda_{3,4} ^2  = \frac{- \alpha - \sqrt{\alpha^2 - 4 \beta}}{2} \quad . 
\end{equation}
The equilibrium point $\mathbf{x}_0=-\mathbf{k}$ is elliptic when all  $\lambda^2_i$, $i=1,\ldots,4$ are negative.
It can be easily seen that $\lambda_{1,2} ^2 <0$ if and only if $\alpha >0$ and $\beta >0$. Under this condition, $\lambda_{3,4} ^2$  is always negative.

\end{proof}

\begin{figure}
    \centering
    \includegraphics[width=0.35\linewidth,
    trim=30 20 -7 20,clip]{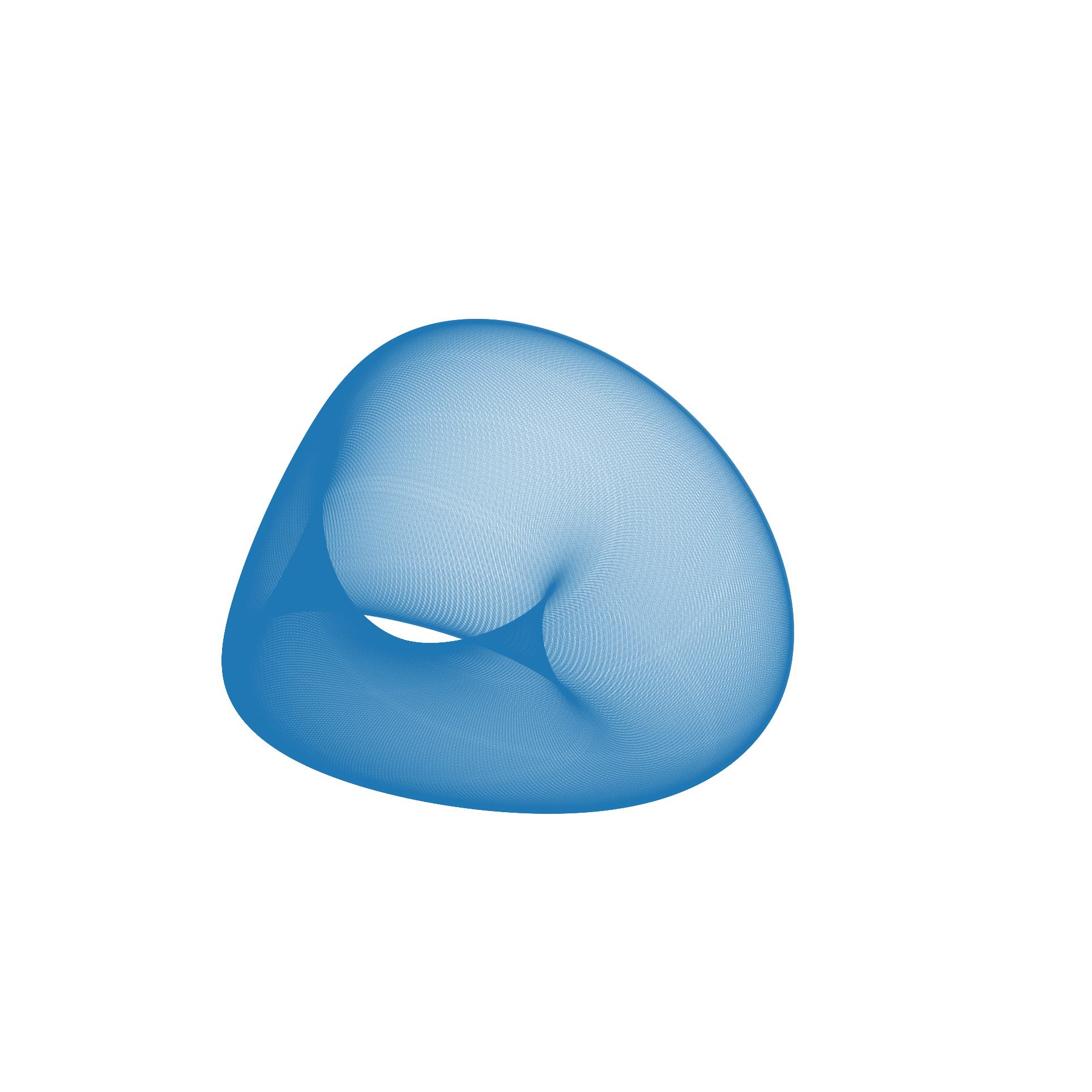}%
    \hspace{-0.11\linewidth}%
    \includegraphics[width=0.35\linewidth,
    trim=30 20 1 20,clip]{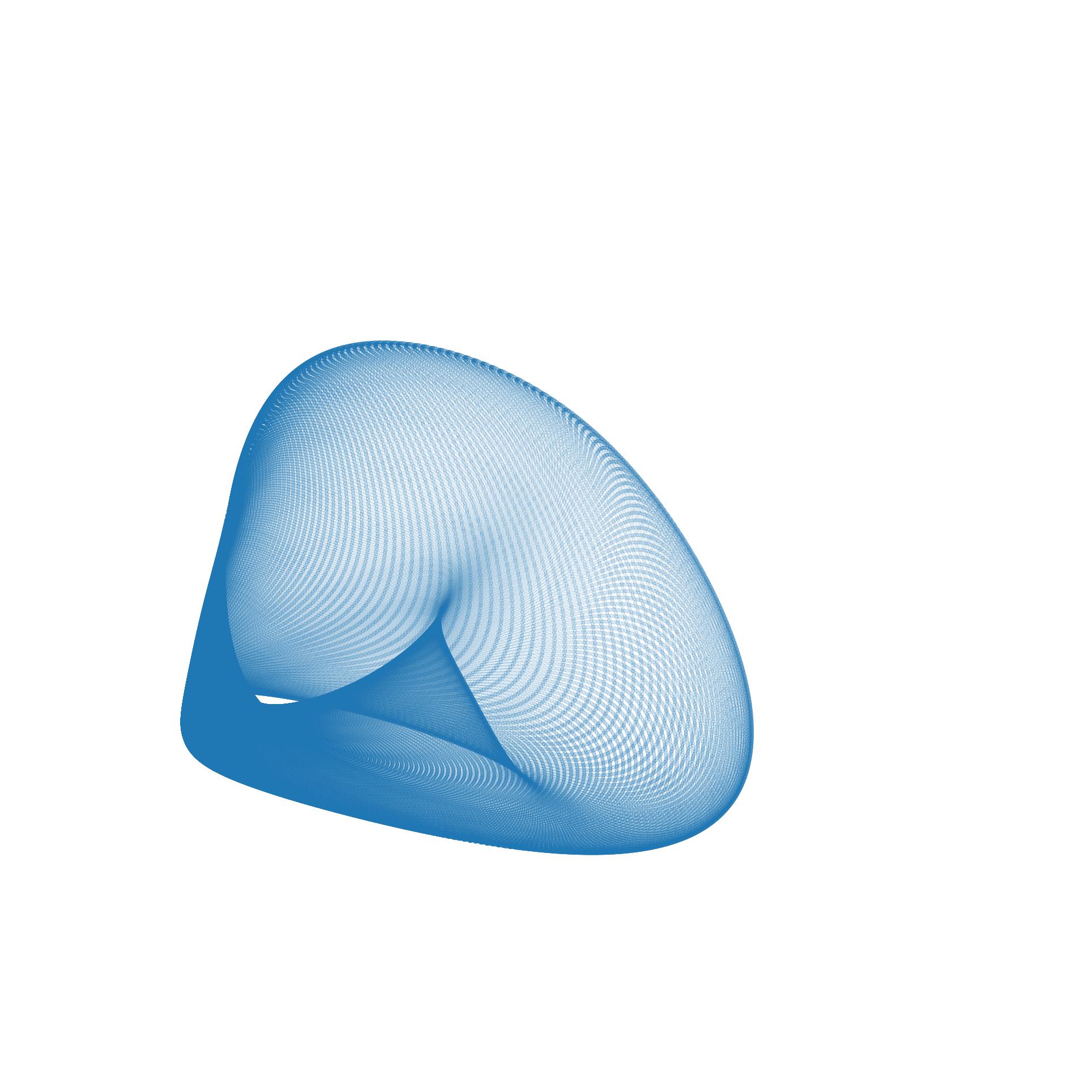}%
    \hspace{-0.11\linewidth}%
    \includegraphics[width=0.35\linewidth,
    trim=30 20 30 20,clip]{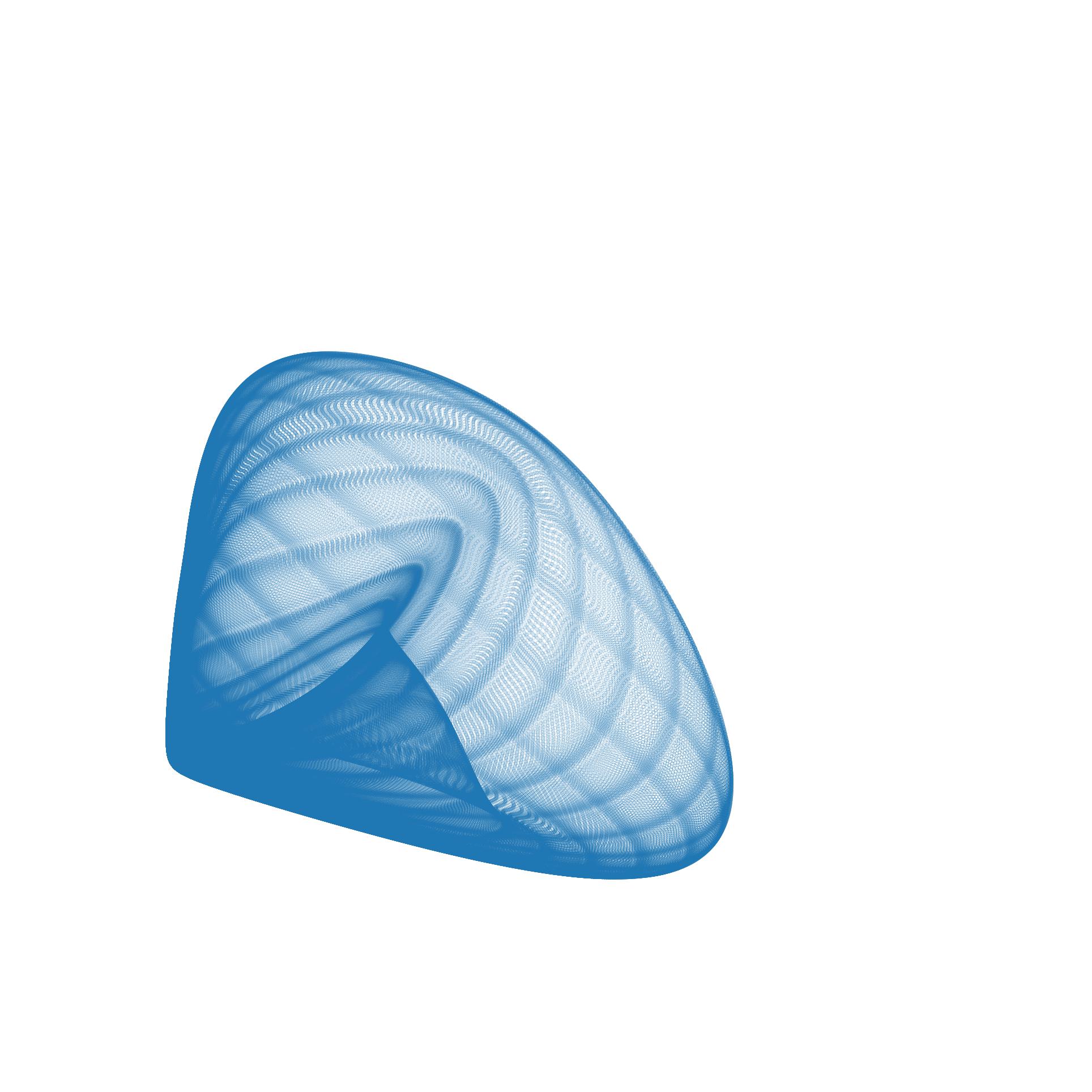}%
    \hspace{-0.11\linewidth}\\

    \vspace{-18pt}

    \includegraphics[width=0.35\linewidth,
    trim=30 20 30 20,clip]{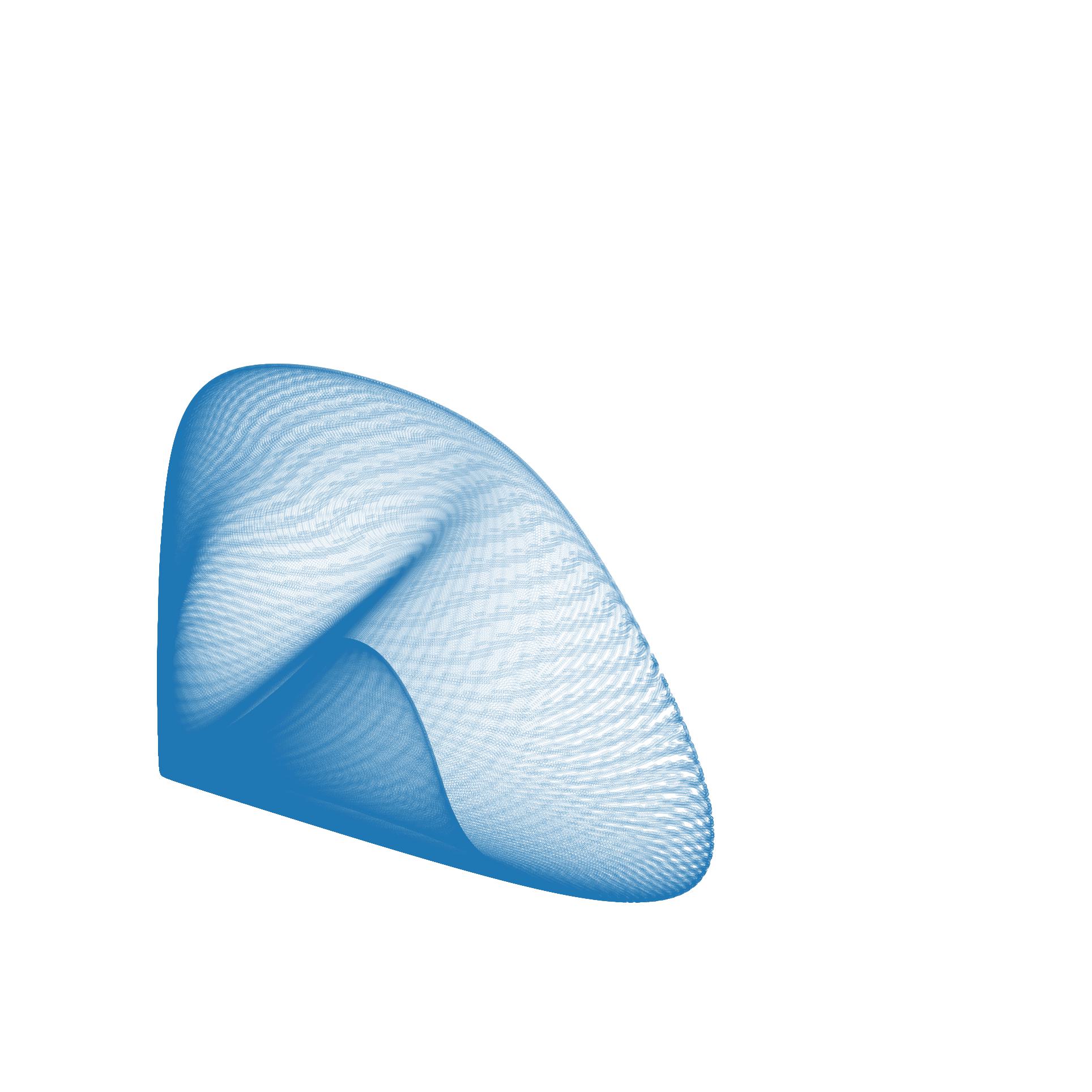}%
    \hspace{-0.11\linewidth}%
    \includegraphics[width=0.35\linewidth,
    trim=30 20 30 20,clip]{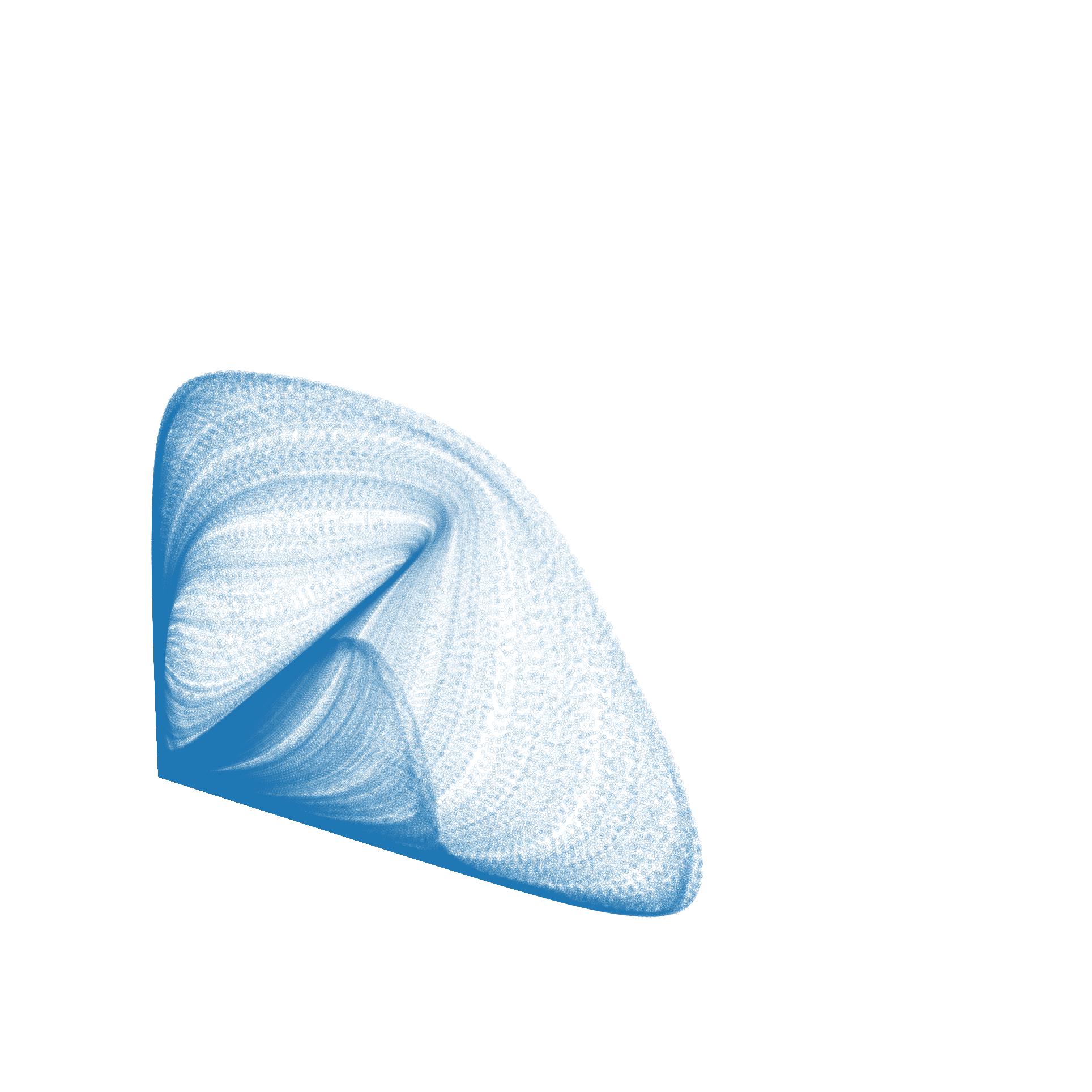}%
    \hspace{-0.11\linewidth}%
    \includegraphics[width=0.35\linewidth,
    trim=30 20 30 20,clip]{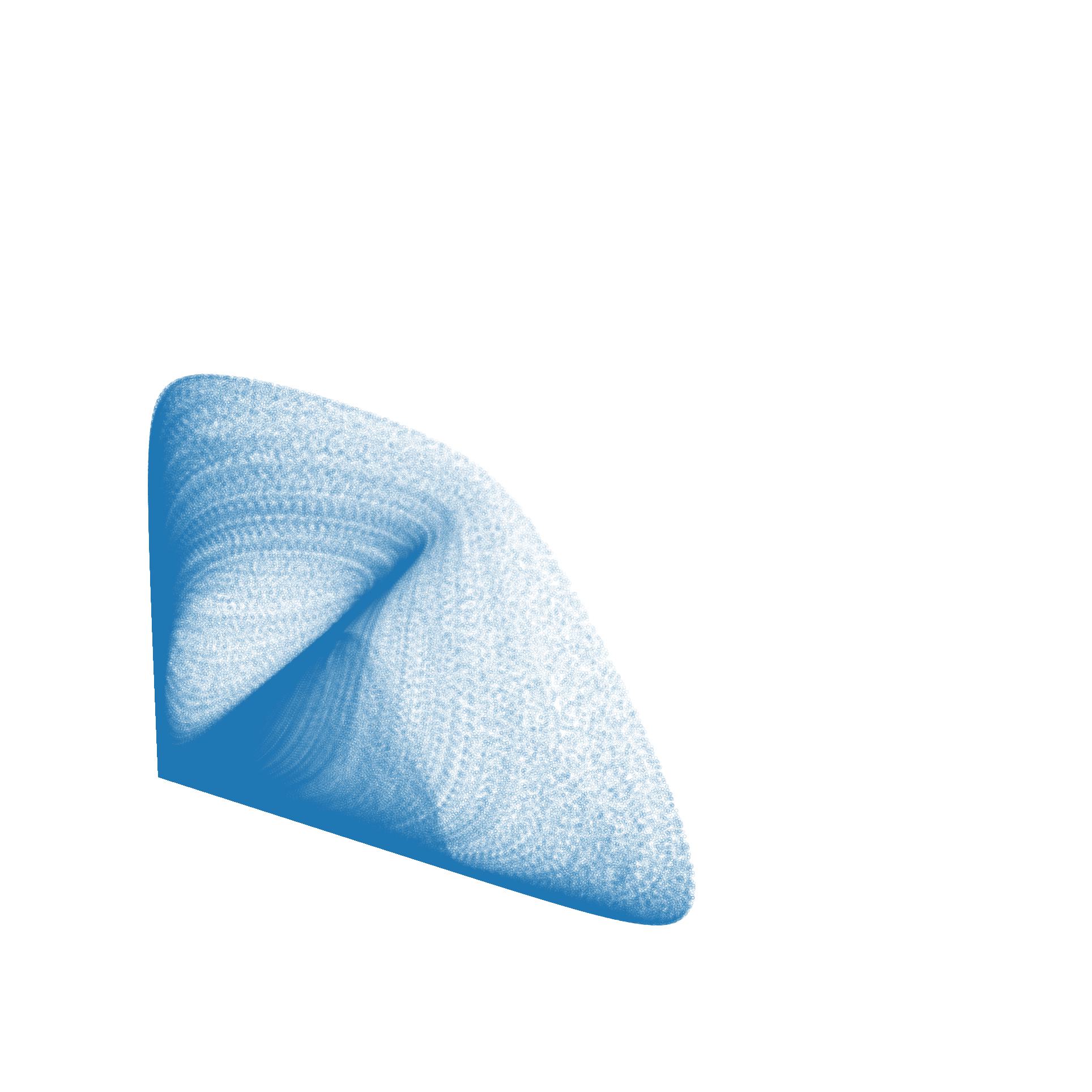}

    \caption{3D projections on the $x_1x_3x_4$ hyperplane as the energy is gradually increased.}
    \label{fig2}
\end{figure}

We note that, whenever the system possesses an elliptic equilibrium, it is the unique interior equilibrium, and requires that $k_i\neq 0$. If any of the parameters $k_i$ vanishes, then $\beta=0$, and the eigenvalues \eqref{eig2} of the Jacobian matrix $Df|{x_0}$ reduce to
\[
\lambda_{1,2}=0, \qquad \lambda_{3,4}=\pm\sqrt{\alpha} \quad.
\]
This indicates that the system does not admit bounded trajectories $\mathbf{x}(t)$ in such cases, and some species either go extinct, as in Figs.\ref{fig1k}(c),(d),  or overgrow. 
More precisely, when $\alpha<0$, the populations $x_3$ and $x_4$ remain balanced, whereas when $\alpha>0$, the equilibrium is a saddle point in the $(x_3,x_4)$--plane, with one population tending to zero and the other growing without bound.


Figure  \ref{fig2} presents 3D projections of a representative orbit around the equilibrium point $\mathbf{x}_0=-\mathbf{k}$ as Proposition 1 describes. We regard negative values for all $k_i$, hence $\alpha, \beta >0$ and $\lambda_i$ are all imaginary. As the energy increases, the topology of the orbit changes to different patterns. 
%

\section{Integrable  and non-integrable cases }
\subsection{Integrable cases \label{integrablecases}}
\begin{figure}
    \centering
    \includegraphics[width=0.5\linewidth, 
    trim=10 50 10 60, clip]{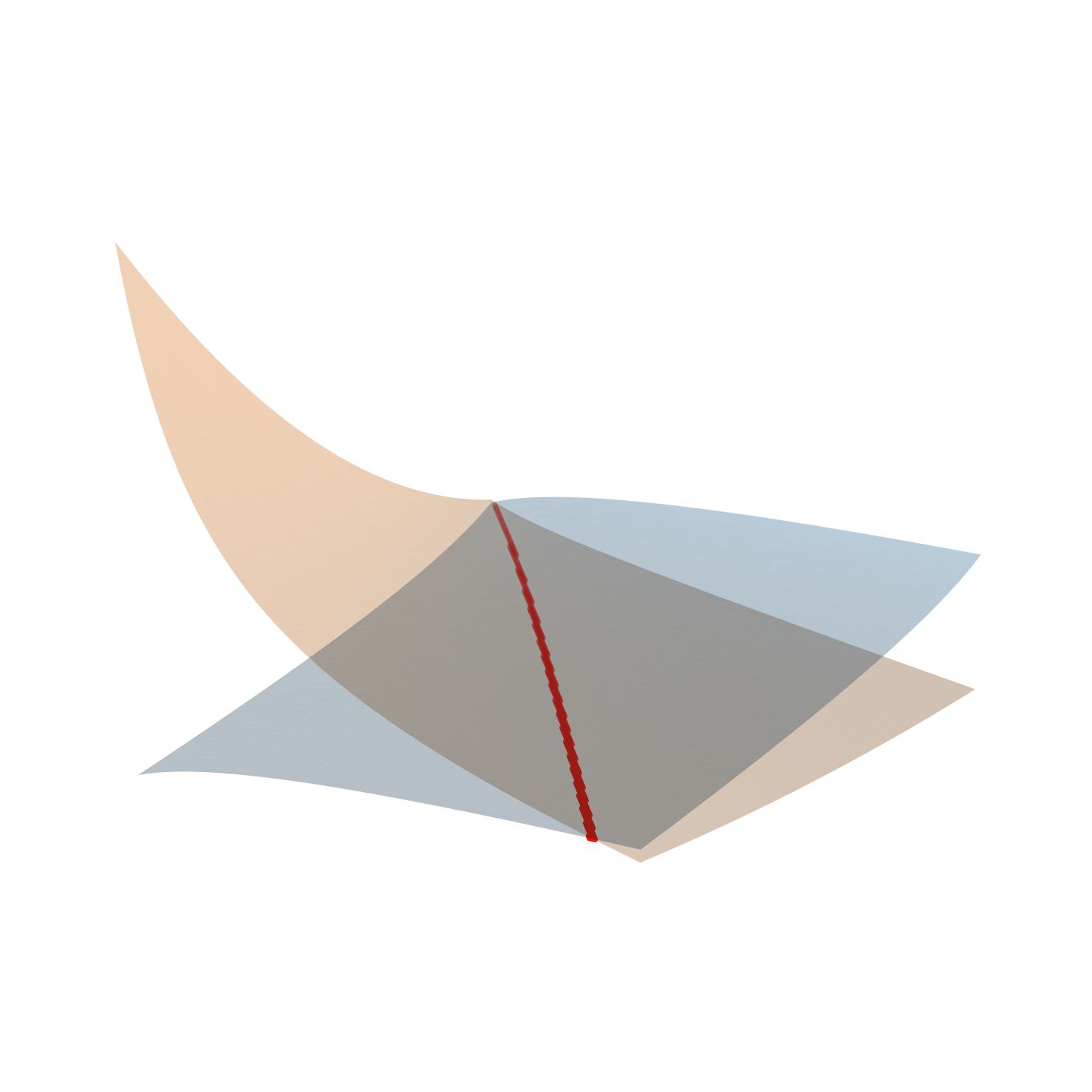}
    \caption{A schematic 3D projection of the intersection (red line) of the two integrals 
    $ I = I_1I_3 = (x_2 + x_3) (x_3 + x_4) /({x_2x_4})$ (light blue) and $H=x_1+x_2+x_3+x_4 + 3\log x_1 + 2 \log x_2 - 2 \log x_3 + 2 \log x_4$ (beige) for $x_1=0.8, x_4=1.2$, $x_2,x_3 \in [0.4, 3]$. }
    \label{fig:placeholder}
\end{figure}

\begin{table}[h]
\centering
\begin{tabular}{|c c c c|}
\hline
\textbf{Cases} & \textbf{Linear Coefficient} & \textbf{Coefficient in Hamiltonian} &\textbf{Integral} \\[0.2cm]
$A$ & $\mathbf{r}=(k_4,-k_1 + k_4,-k_1 + k_4,-k_1) $ & $   \mathbf{k}=(k_1,0,0,k_4)$ & $I_1=(x_2+x_3)\frac{1}{x_1x_4}$ \\[0.2cm]
$B_1$ & $\mathbf{r}=(k_2,-k_1,-k_1 - k_2,-k_1 - k_2) $ & $\mathbf{k}=(k_1,k_2,0,0)$ & $I_2=  (x_3+x_4) \frac{x_1}{x_2} $ \\[0.2cm]
$B_2$ & $ \mathbf{r}=(k_3 + k_4,k_3 + k_4,k_4,-k_3) $ & $   \mathbf{k}=(0,0,k_3,k_4)$  & $I_3 = (x_1+x_2) \frac{x_4}{x_3}$ \\[0.2cm]
$C_1$ &  $\mathbf{r}=(k_2,-k_1,-k_1,-k_1) $ & $   \mathbf{k}=(k_1,k_2,-k_2,k_2)$ & $I_4=I_1 I_2 - 1=(x_2+x_3+x_4) \frac{x_3}{x_2 x_4}$ \\[0.2cm]
$C_2$ & $\mathbf{r}=(k_4,k_4,k_4,-k_1) $ & $ \mathbf{k}=(k_1,-k_1,k_1,k_4)$ & $ I_5 = I_1 I_3 -1 =(x_1+x_2+x_3) \frac{x_2}{x_1 x_3}$ \\[0.2cm]
$D$ & $\mathbf{r}=(-k_1,-k_1,-k_3,-k_3) $ & $ \mathbf{k}=(k_1,-k_1,k_3,-k_3)$ & $I_6 = I_2^{k_1} \cdot I_3^{- k_3} $ \\[0.3cm]
\hline
\end{tabular}
\caption{Integrable cases of the four-dimensional Lotka--Volterra model (\ref{oursystem}). Each case consists of a two-parametric family.}
\label{table2}
\end{table} 
The integrability and solutions of the system (\ref{themodel}) for arbitrary dimension $n$ and various choices of the parameters $a_i, r_i, i=1, \ldots, n$ have been investigated in a series of papers \cite{HK2019,  KQV, KKQTV, VKMcQR2021, VKMcQR2025}. In particular, for $r_i=0$ and any $a_i \in \mathbb{R}$, the system is always Liouville integrable and superintegrable admitting $n-1$ functionally independent first integrals. For the case of $n=4$, $a_i=1, r_i=0$, these integrals can be expressed as 
\[
I_1=(x_2+x_3)\frac{1}{x_1x_4}, \quad
I_2 = (x_3+x_4) \frac{x_1}{x_2}, \quad 
I_3 = (x_1+x_2) \frac{x_4}{x_3} \, ,
\]
and the Hamiltonian simplifies to $$H_0=x_1 +x_2+ x_3 + x_4 = (I_1-1) I_2 I_3 \,.$$ 
Clearly, any functional combination of the integrals \(I_i\), of the form \(\mathcal{I}=\Phi(I_1,I_2,I_3)\), is also an integral in this case.

In the general case of arbitrary \(r_i\), the quantities \(I_1,I_2,I_3\) are not necessarily integrals. Nevertheless, for particular choices of the constants \(r_i\), with $r_1^2 + r_2^2 + r_3^2 + r_4^2 \neq 0$, certain combinations \(\mathcal{I}=\Phi(I_1,I_2,I_3)\) remain integrals of (\ref{oursystem}).
More precisely, by  evaluating the Poisson bracket of \(\mathcal{I}=\Phi(I_1,I_2,I_3)\) with the Hamiltonian (\ref{hamilton2}), we obtain  the  expression
$$
 \{ \mathcal{I}, H \} = \sum_{i=1}^4 \{ \Phi(I_1,I_2,I_3), k_i \log {x_i} \}=  
 \frac{k_3 x_2 - k_2 x_3}{x_2 + x_3}  I_1 \frac{\partial \Phi}{\partial I_1} + \frac{k_4 x_3 - k_3 x_4}{x_3 + x_4}  I_2 \frac{\partial \Phi}{\partial I_2} +
\frac{k_2 x_1 - k_1 x_2}{x_1 + x_2} I_3 \frac{\partial \Phi}{\partial I_3} \, ,
$$
from which  the following Lemma yields.  

\begin{Lemma}
 If  
\begin{equation} \label{intLem}
 \frac{k_3 x_2 - k_2 x_3}{x_2 + x_3}  I_1 \frac{\partial \Phi}{\partial I_1} + \frac{k_4 x_3 - k_3 x_4}{x_3 + x_4}  I_2 \frac{\partial \Phi}{\partial I_2} +
\frac{k_2 x_1 - k_1 x_2}{x_1 + x_2} I_3 \frac{\partial \Phi}{\partial I_3} =0\,,
\end{equation}
for a differentiable function 
$\Phi : D \subset\mathbb{R}^3 \rightarrow \mathbb{R}$, then  
$\mathcal{I}=\Phi(I_1,I_2,I_3)$ is a first integral of system (\ref{oursystem}).
\end{Lemma}
For arbitrary $k_i$, the generic solution of equation (\ref{intLem}) is just the constant function $\Phi(I_1,I_2,I_3)=c$. However, particular choices of $k_i$ lead to non-trivial integrals. For example, when $k_3=k_4=0$,  equation (\ref{intLem})  reduces to 
$$ \frac{  - k_2 x_3}{x_2 + x_3}  I_1 \frac{\partial \Phi}{\partial I_1} + 
\frac{k_2 x_1 - k_1 x_2}{x_1 + x_2} I_3 \frac{\partial \Phi}{\partial I_3} =0 \, ,$$
which admits the simple solution $\Phi(I_1,I_2,I_3) = I_2$. Similarly, for $k_2=-k_1,k_3 = k_1$ equation (\ref{intLem}) simplifies to   
$$ k_1 I_1 \frac{\partial \Phi}{\partial I_1} + \frac{k_4 x_3 - k_1 x_4}{x_3 + x_4}  I_2 \frac{\partial \Phi}{\partial I_2} - k_1 I_3 \frac{\partial \Phi}{\partial I_3} =0\, , $$
and admits the solution  $\Phi(I_1,I_2,I_3) = I_1 I_3$.
In Table \ref{table2}, we present  integrable cases that follow from Lemma 1 for various choices of the parameters $k_i$. These cases have also been derived in \cite{VKMcQR2021,VKMcQR2025} using the theory of Darboux polynomials. 
All the intagrable cases in Table \ref{table2} are two-parametric with two independent parameters $k_i, k_j$, $i \ne j$.
We note that the integrable cases \(B\) and \(C\) occur in pairs due to the permutation symmetry of the system
$(x_1,x_2,x_3,x_4) \mapsto (-x_4, -x_3, -x_2, -x_1) $ and $(k_1,k_2,k_3,k_4)\mapsto(-k_4,-k_3,-k_2,-k_1)$.
Under this symmetry transformation, case \(B_1\), with associated integral \(I_2\), corresponds to case \(B_2\), with integral \(I_3\). Likewise, case \(C_1\), with \(I_4 = I_1 I_2 -1\), corresponds to case \(C_2\), with \(I_5 = I_1 I_3 -1\), while the integrals \(I_1\) and \(I_6\), associated with cases \(A\) and \(D\), respectively, remain invariant.

We remark that Lemma 1 may not exhaust all possible integrable cases. Instead, it identifies those cases admitting integrals of system (\ref{oursystem}) that are also integrals of the corresponding system with all $r_i=0$.  
The existence of additional integrable cases that do not possess this property, and therefore lie outside Lemma 1's framework, is left for future investigation.

It is worth emphasizing that all identified $n=4$ integrable families inherently fail to satisfy the necessary criteria for an interior elliptic equilibrium set by Proposition 1 ($\alpha>0$ and $\beta>0$).  
In these cases, there are  two possible asymptotic behaviors, either exponential divergence or collapse onto the extinction planes for some of the populations $x_i$. As a consequence, exact integrability within this hierarchical predator-prey network leads to ecological collapse. 

In the following subsection we focus on cases that exhibit chaotic behavior.

\begin{figure}
    \centering
            \includegraphics[width=0.4\linewidth]{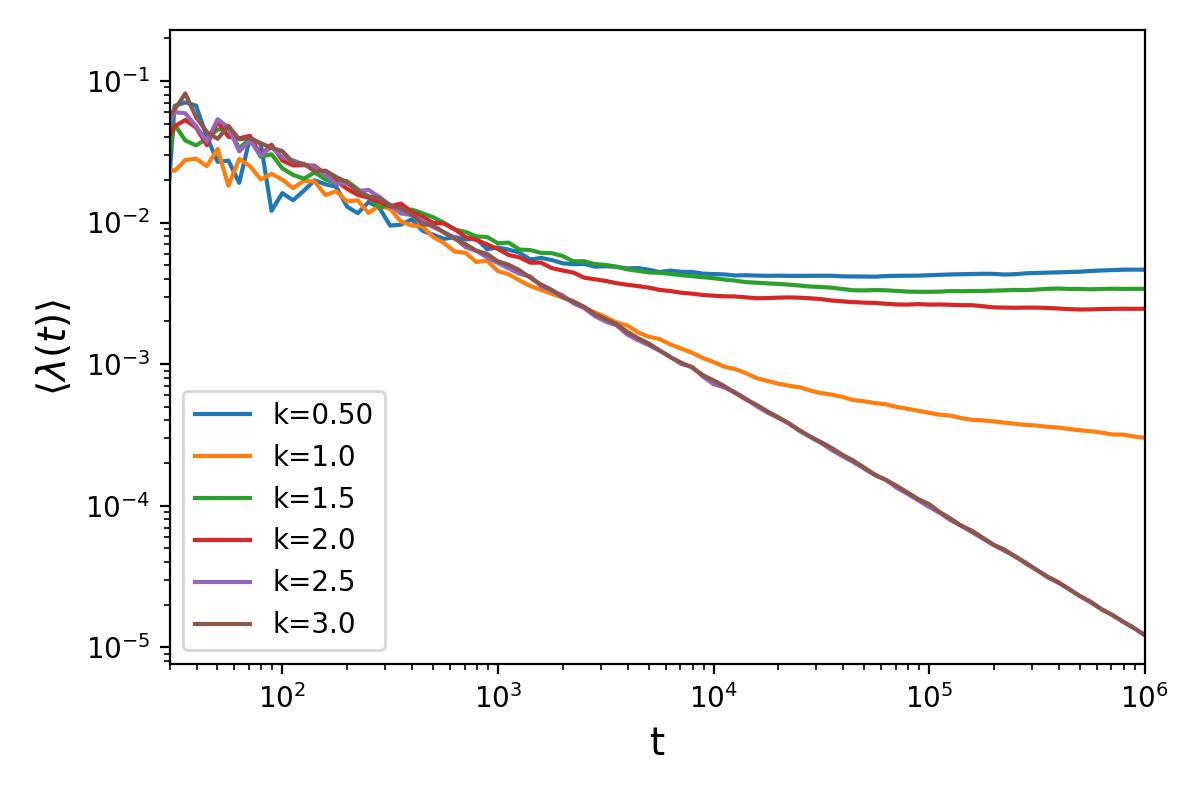}
        \includegraphics[width=0.4\linewidth]{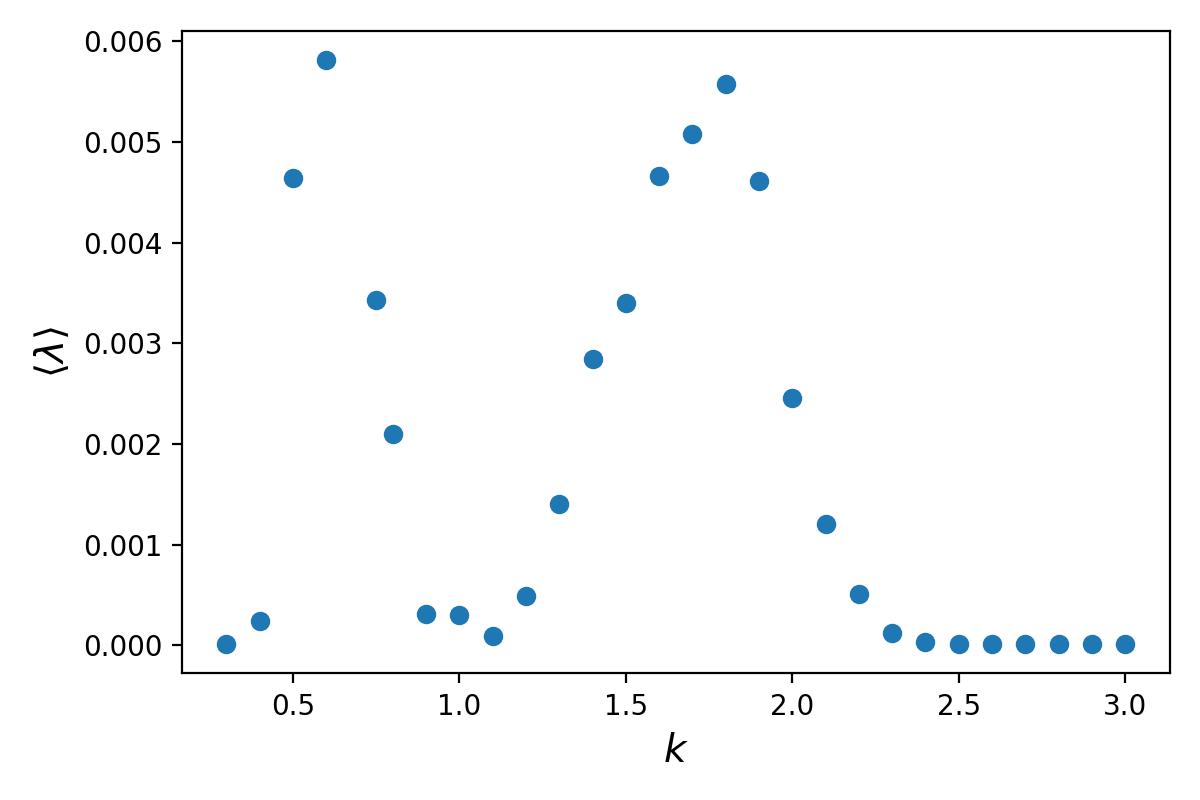}
        \caption{ Left: The evolution of the averaged finite-time Lyapunov exponent $< \lambda (t) > = \sum_{m=1}^{100} \lambda_m (t) /100$ for a selected number of $k$ values. 
        Right: The converged values of the averaged maximum Lyapunov exponents $< \lambda > $ versus $k$.   
        \label{LyapExp}}
\end{figure}

\subsection{Non-integrable cases\label{chaoscases}}
\begin{figure}
    \centering
        \includegraphics[width=0.4\linewidth]{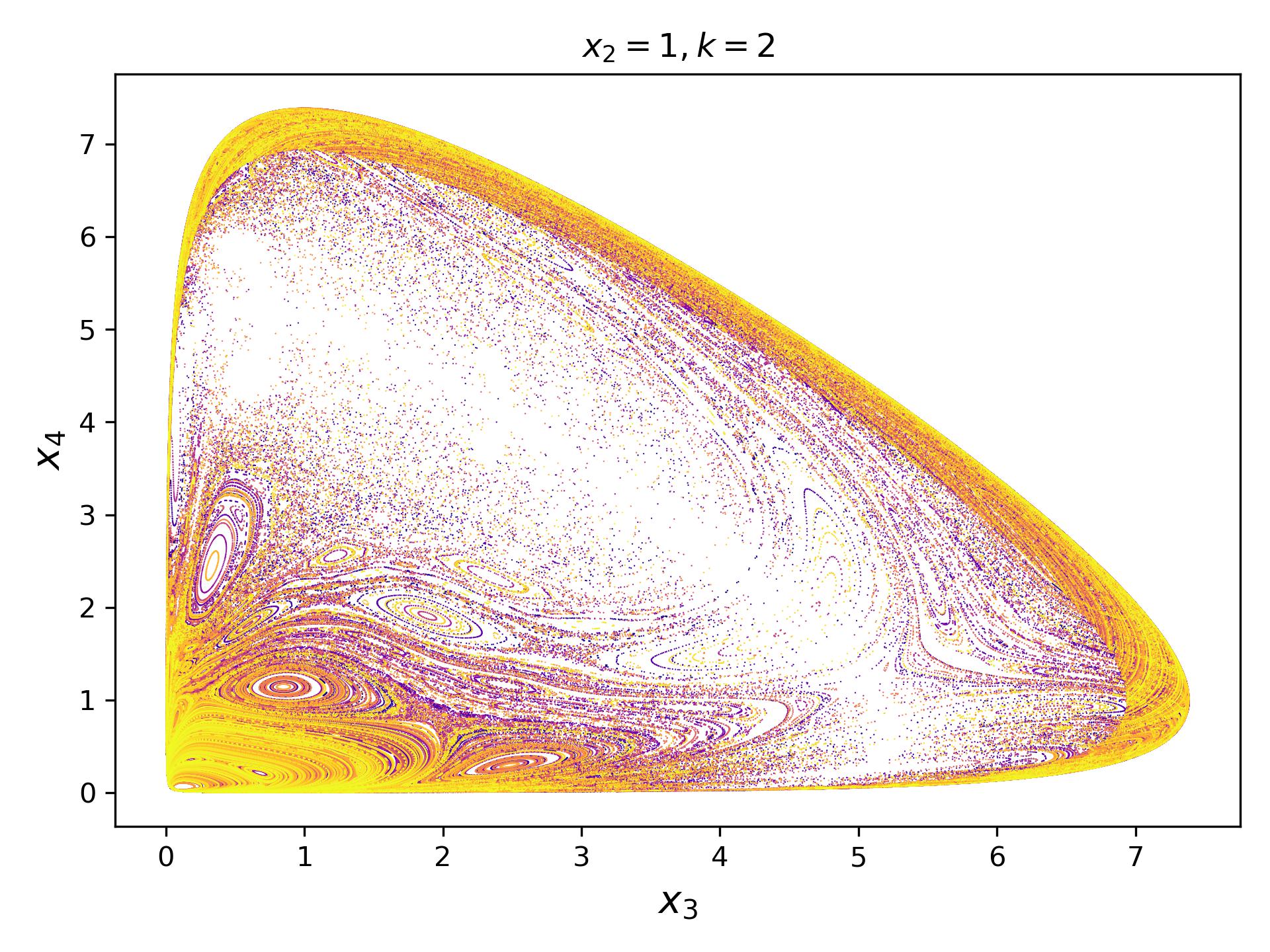}
        \includegraphics[width=0.4\linewidth]{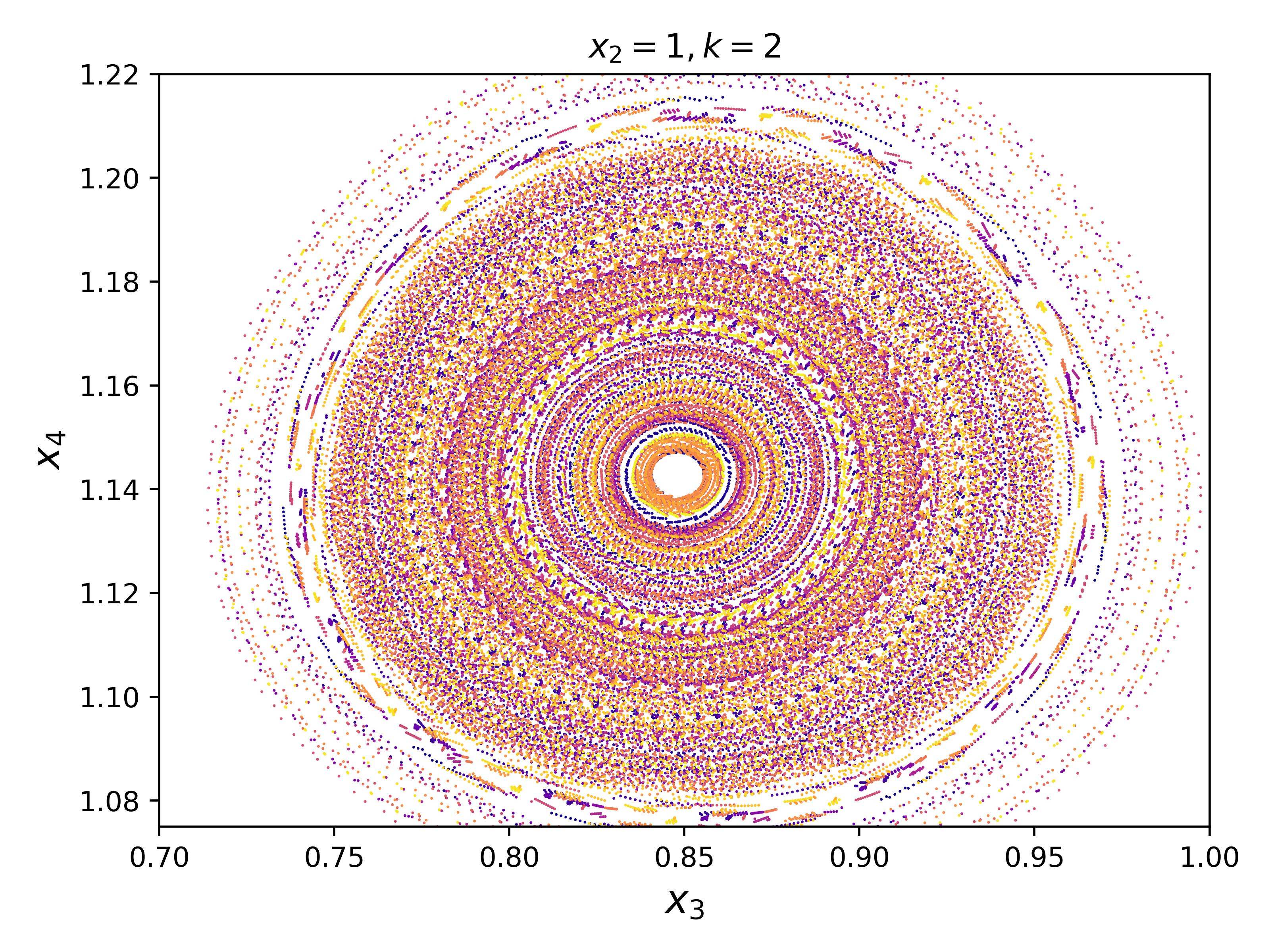}
        \includegraphics[width=0.4\linewidth]{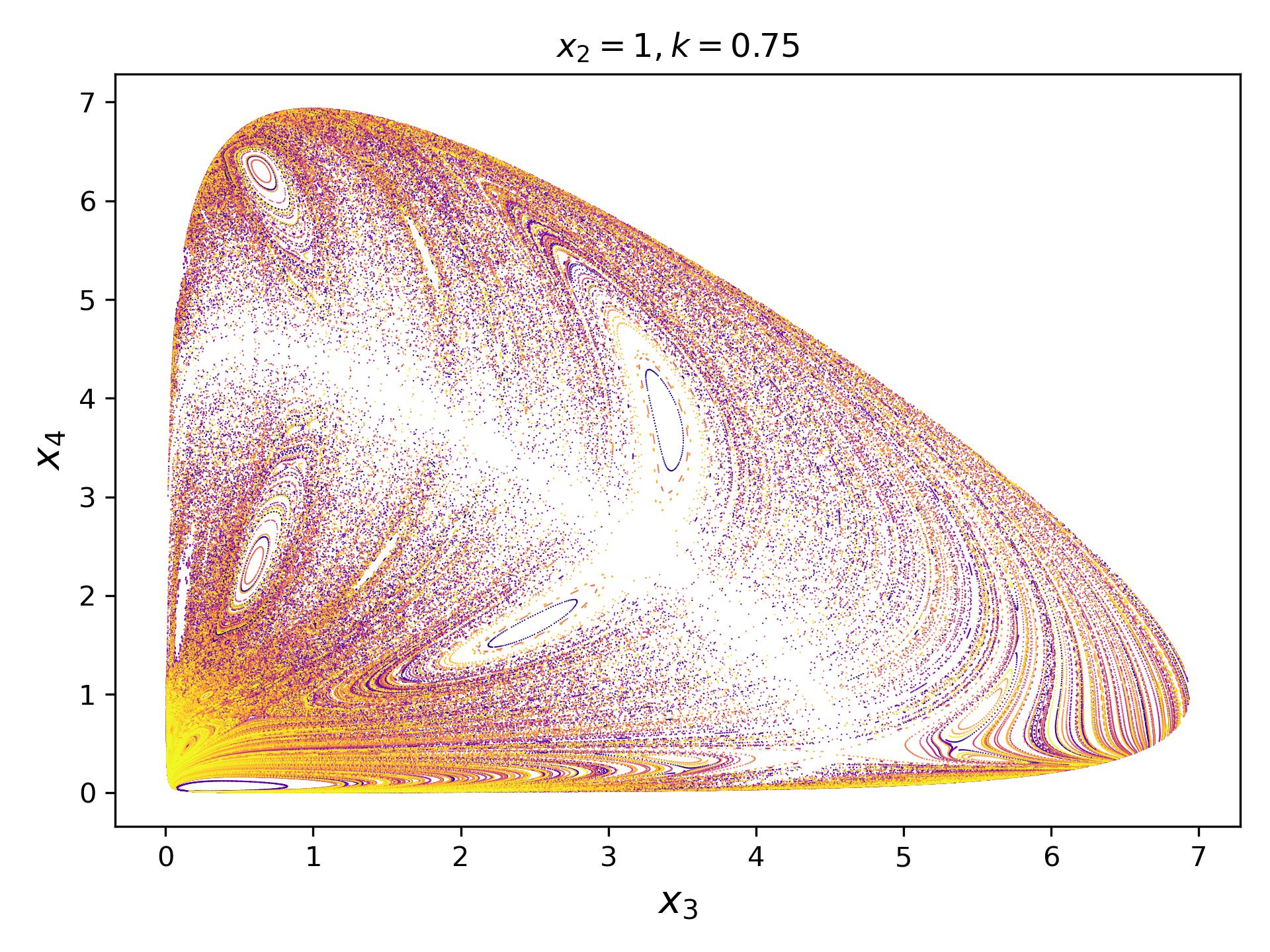}
        \includegraphics[width=0.4\linewidth]{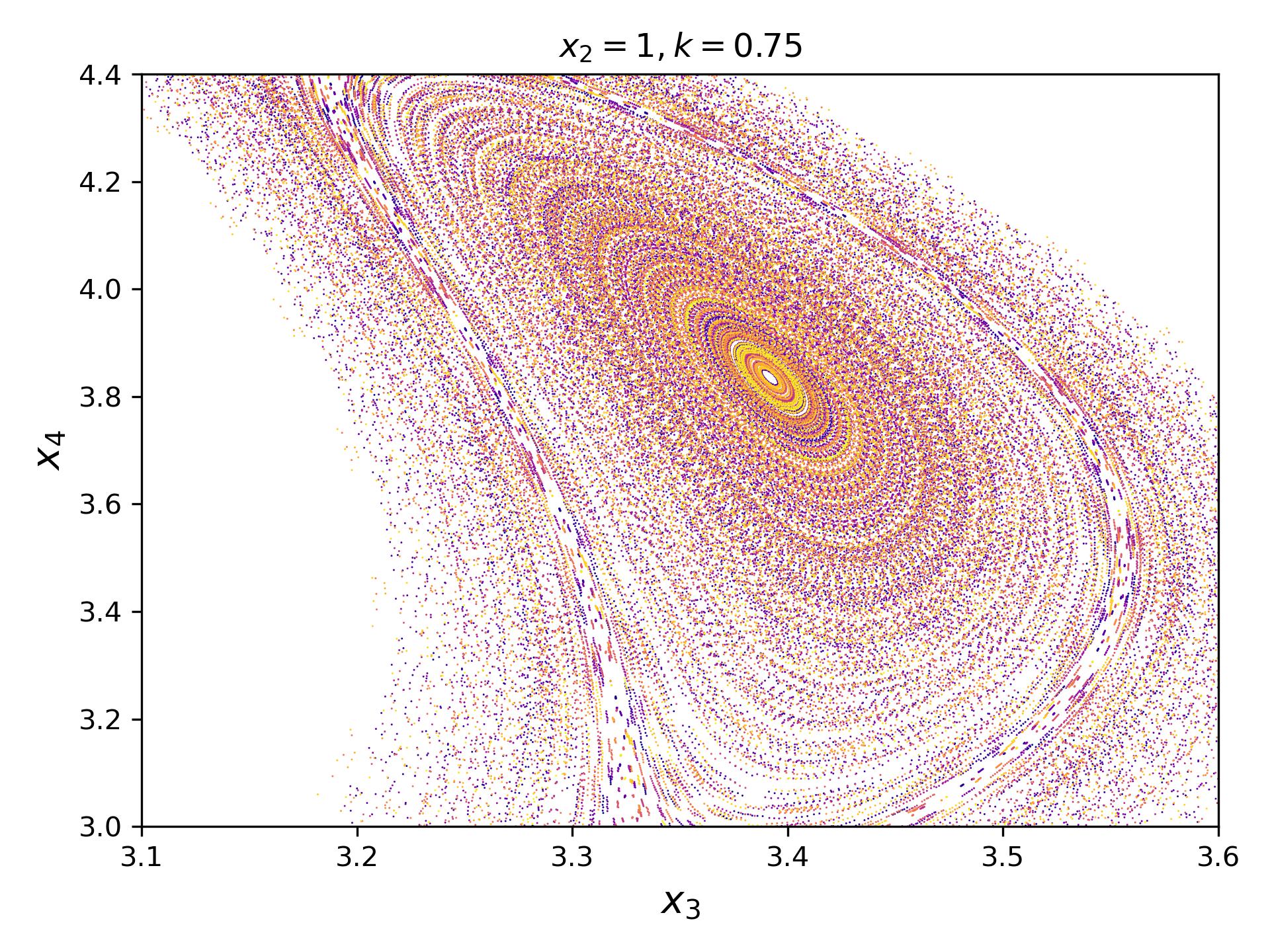}
        \caption{Poincar{\'e} surfaces of section of $x_2=1$, $x_1>1$ displaying $x_3$, $x_4$ variables. 
        Upper left: the section for $\mathbf{k} = (-2, -2, -1 ,-1)$, upper right: a closed look around $(0.85, 1.15)$ showing organized behavior,
        lower left: the section for $\mathbf{k} = (-0.75, -0.75, -1 ,-1)$, lower right: a closer look around $(3.4, 3.9)$. \label{poincare1}   }
\end{figure}

\begin{figure}
    \centering
    \includegraphics[width=0.4\linewidth]{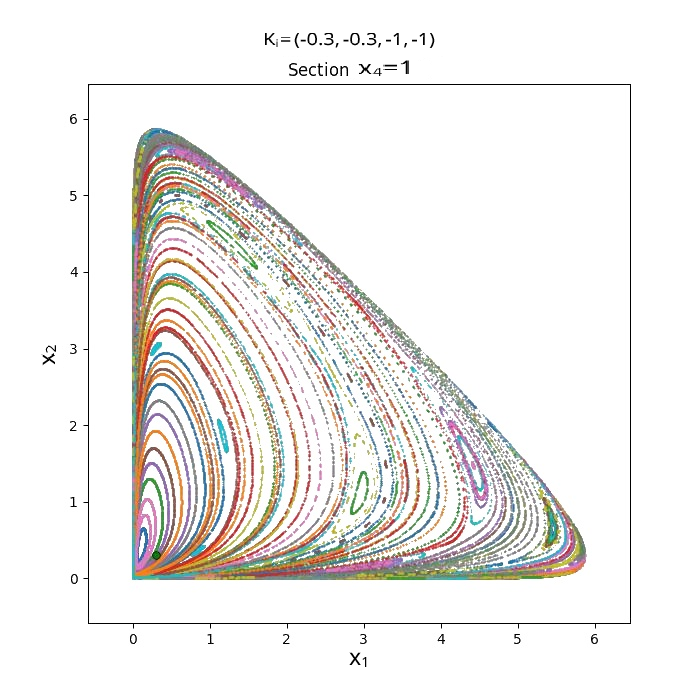}\hspace{-0.04\linewidth}
    \includegraphics[width=0.4\linewidth]{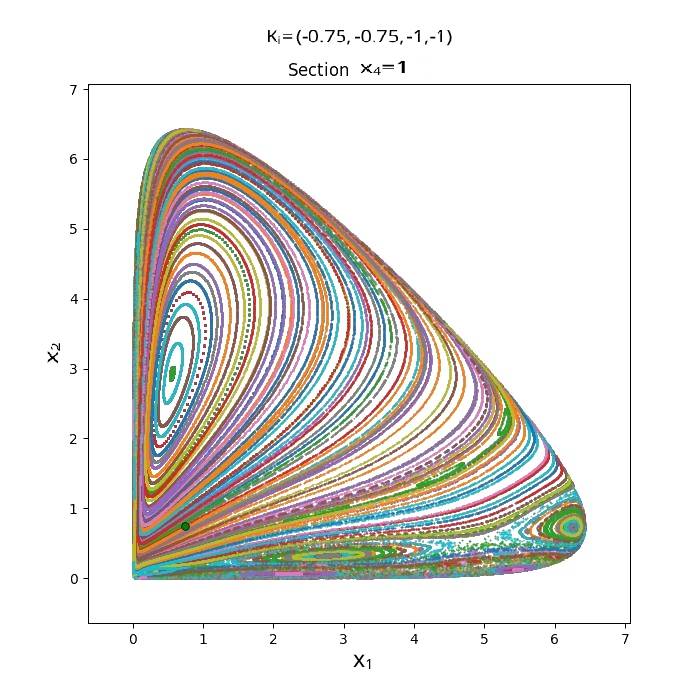}\\
    \includegraphics[width=0.4\linewidth]{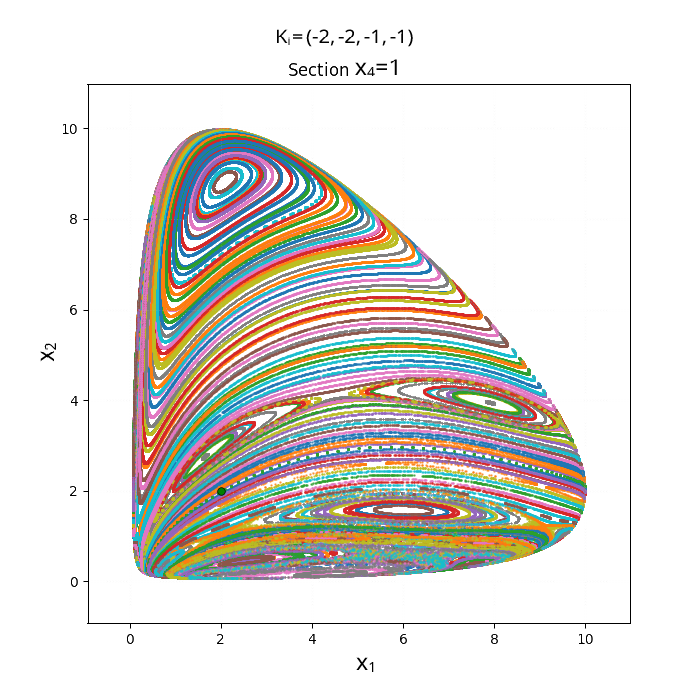}\hspace{-0.04\linewidth}
    \includegraphics[width=0.4\linewidth]{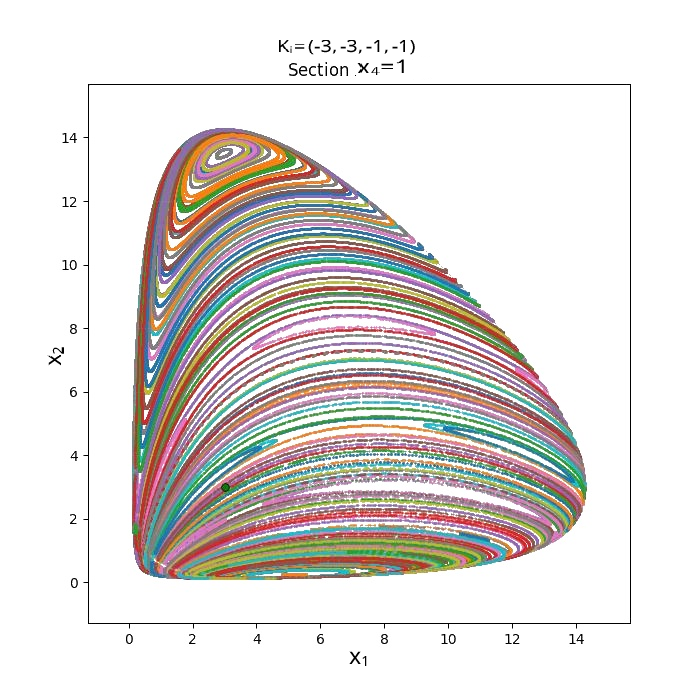}
    \caption{   Poincar{\'e} surfaces of section defined by $x_4=1$, $x_3>1$ displaying $x_1,x_2$ variables. 
    Upper left: $\mathbf{k} = (-0.3, -0.3, -1 ,-1)$, upper right: $\mathbf{k} = (-0.75, -0.75, -1 ,-1)$.
        Lower left: $\mathbf{k} = (-2, -2, -1 ,-1)$, lower right: $\mathbf{k} = (-3, -3, -1 ,-1)$. $E=8$ in all four panels.  Different colors highlight the trajectories obtained from distinct initial conditions.
    \label{last}}
\end{figure}

Let us first study cases that display chaotic behavior through two different methods: the calculation of the maximum Lyapunov exponent and the Poincar{\'e} surfaces of section.  To make the comparison more consistent, we fix the total energy of the system to $E=8$, and vary only the values of the  $\mathbf{k}$ parameters for different random initial conditions. In particular, we start with the cases that produce elliptic equilibrium points, while $x_i>0$, yielding that $k_i<0$. To simplify the calculations, we set $\mathbf{k} = (-k, -k, -1 ,-1)$, and vary the values of a single arbitrary parameter $k$, starting from a positive value and approaching zero (integrable case).

To calculate the maximum Lyapunov exponent $\lambda$, we numerically solve the system's variational equations 
alongside those of the system (\ref{oursystem}) (reference orbit). Since the model is low-dimensional and its fixed energy $E=8$ is not high enough, there is a plethora of cases yielding either too small or zero Lyapunov exponents. 
To tackle this issue, we consider $m=100$ random initializations for each $k$ value and evaluate the average value $<\lambda>$ over the $100$ converged maximum Lyapunov exponents. 
Hence, $<\lambda>$ indicates also a percentage of chaotic and non-chaotic orbits in phase space.

In Fig.\ref{LyapExp}(a) we display the evolution of the finite-time, averaged over 100 realizations maximum Lyapunov exponents $< \lambda(t) >$  for a few selected $k$ values in $\mathbf{k} = (-k, -k, -1 ,-1)$. We thus observe that for the values $k=0.5, 1.5, 2$, $< \lambda (t)> $ has converged to a positive constant. Instead, for $k=2.5, 3$
 $< \lambda (t) > $ continues to decay like $1/t$ and chaos remains undetectable within the time-window up to $10^6$,  signifying more organized behavior at higher $k$ values. The right panel of  Fig.\ref{LyapExp}  shows the values of $< \lambda >$ at time $10^6$ for $k$ ranging from $0.1$ to $3$. It is thus observed that for $k$ values within the interval $[0.3, 2.4]$ these have converged to a positive constant, whereas for values within $[2.5, 3]$ the finite-time Lyapunov exponent still decays as $1/t$.  
  An interesting remark, at this point, is that the maximum Lyapunov exponents vary non-monotonically with $k$, showing two local maxima around $k=0.6$ and $1.8$, and minima around $k=0.3, 1.1$. At higher $k$ values, where $k \ge 2.5$, chaos is undetectable.

Fig. \ref{poincare1} presents the two Poincar{\'e} surfaces of section defined by $x_2=1$, $x_1>1$ and plotted in the  $(x_3,x_4)$ plane for $k=2$ and $k=0.75$.  Each Poincar{\'e} surface has been evaluated by considering 200 random initial conditions on the corresponding energy surface, with the different colors representing the trajectories obtained from these distinct initial conditions.
Both sections exhibit coexistence of chaos and organized behavior, which is in accordance to the averaged Lyapunov exponent values of Fig. \ref{LyapExp}. It is worth noting that  Poincar{\'e} surfaces of section $x_2=1$, $x_1>1$  show no points along a curve in the plot, which accounts for the white region in the middle of the plot. This phenomenon arises from the structure of the manifold defined by the preimage of the Hamiltonian function (\ref{hamilton}) $H^{-1}(8)$, which does not admit trajectories in this region.  Instead, the  Poincar{\'e} surfaces of section $x_4=1, x_3>1$, appearing in Fig.\ref{last}, show a clearer representation without the manifold exclusion. However, the sections $x_2=1, x_1>1$ in Fig.\ref{poincare1} display a richer phase space structure, with numerous stability islands and denser chaotic regions.

The right panels of Fig. \ref{poincare1} contain a focused regular region 
of the Poincar{\'e} surfaces $x_2=1, x_1>1$ (left panels) and around a stable periodic orbit (central point in the plot). Such orbits wind around tori in the 4-dimensional phase space, such as we see in Fig.\ref{fig2} (a), which is generated by a single trajectory.

Finally, in  Fig.\ref{last} we compare four Poincar{\'e} surfaces of section defined by $x_4=1$, $x_3>1$ and plotted in the  $(x_1,x_2)$ plane for $\mathbf{k} = (-k, -k, -1 ,-1)$, where (a) $k=0.3$, (b) $k=0.75$, (c) $k=2$ and (d) $k=3$.
The choice of $x_4=1$ intersections has the advantage of yielding clearer images, with no excluded regions, unlike the $x_2=1$ case. On the other hand, chaotic regions appear to be  less dense and harder to identify. In Fig.\ref{last}(a) for $k=0.3$ there is an additional difficulty: Low $k$ values produce the `butterfly'  orbits of \ref{fig1k}(b), where trajectories exhibit excursions between two planes and sparsely intersect the  Poincar{\'e} surface of section. Comparing with the rest panels, we observe that the separatrices between consecutive stability islands in the lower parts of Figs.\ref{last}(b),(c) are connected with chaotic domains
containing `scattered' clouds of points. Such regions are quite subtle in Figs.\ref{last}(d), where positive Lyapunov exponents have not been identified.  

\subsection{Non-sustainable solutions \label{newcases}}
 \begin{figure}
    \centering
        \includegraphics[width=0.4\linewidth]{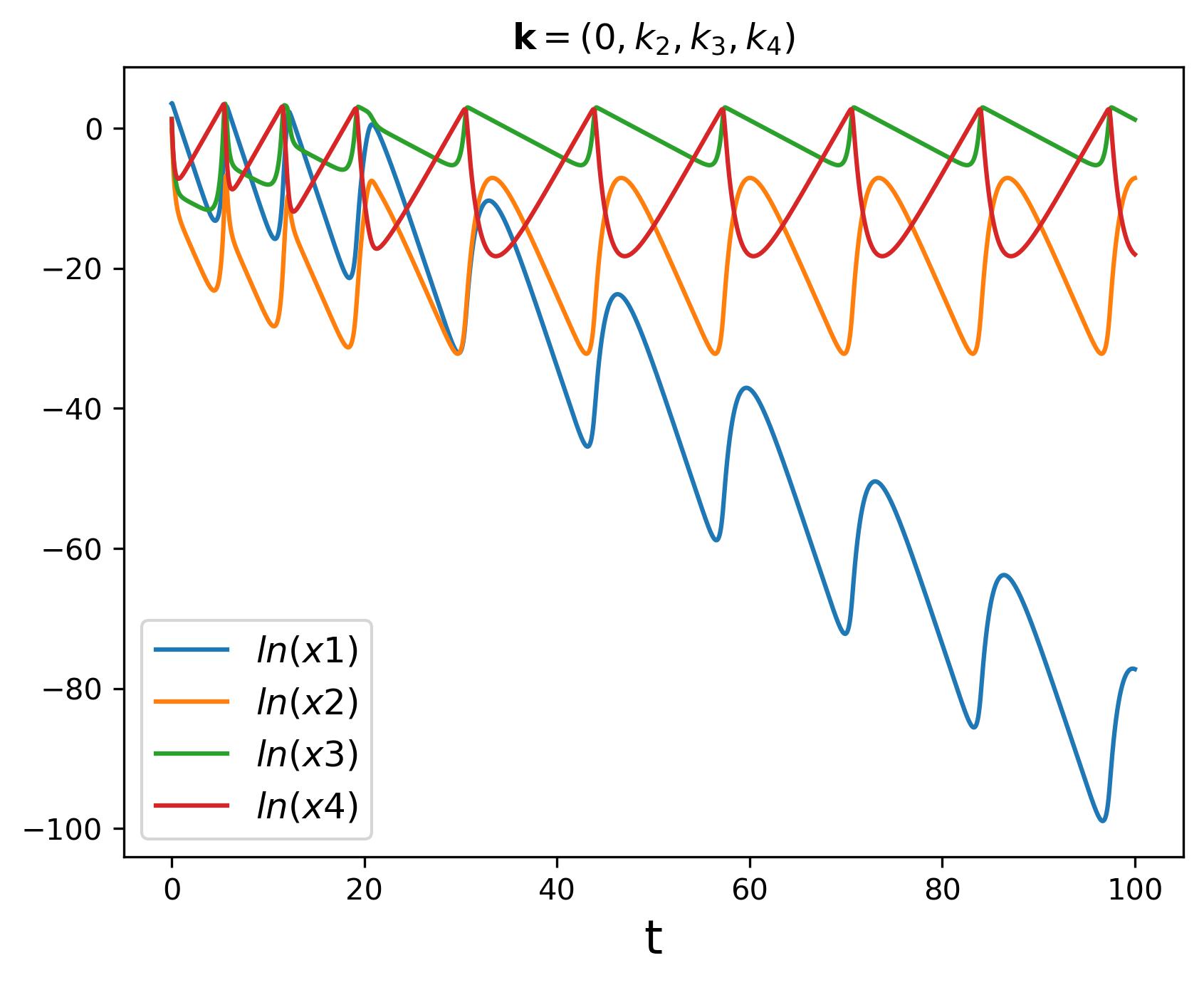}
        \includegraphics[width=0.4\linewidth]{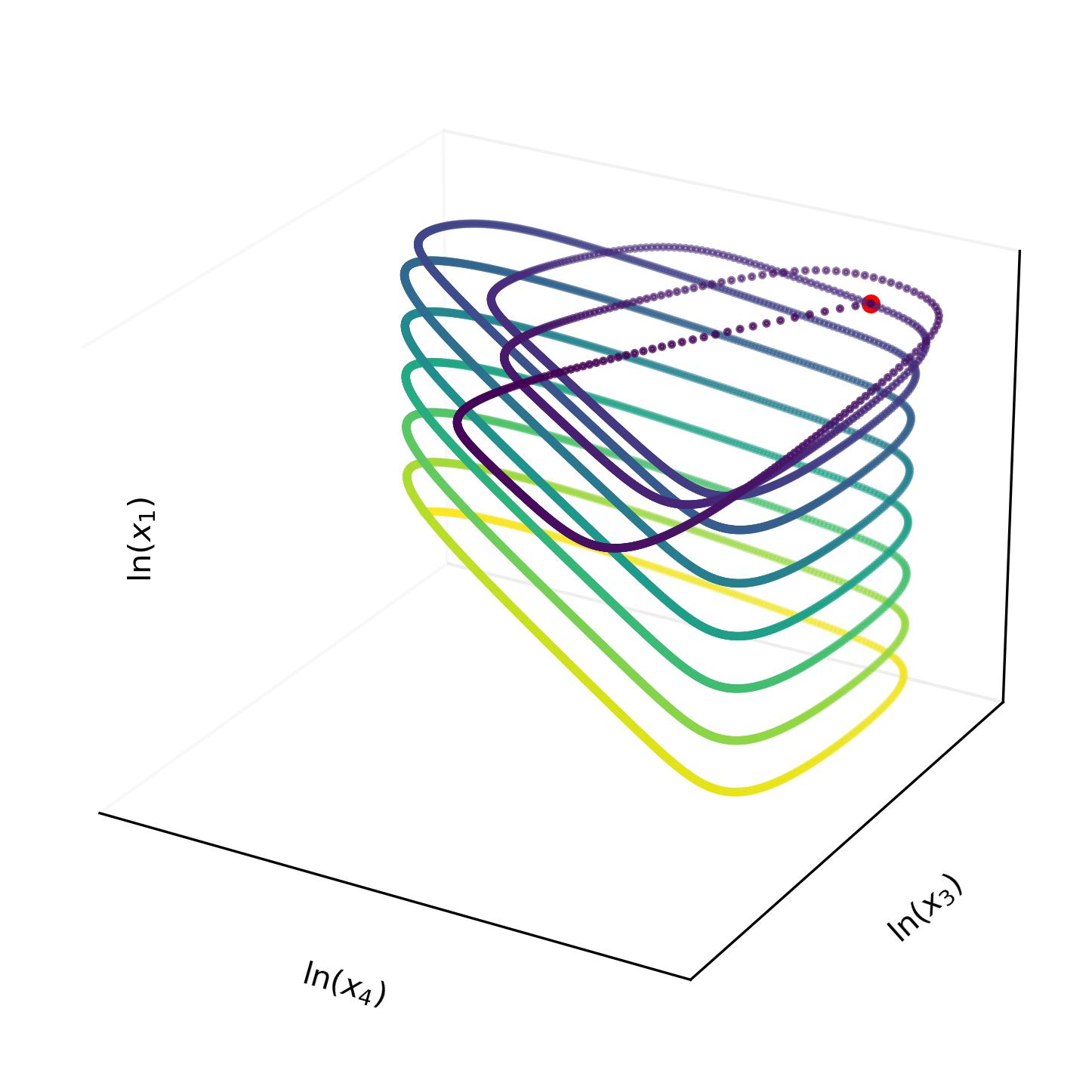}\\
        \includegraphics[width=0.4\linewidth]{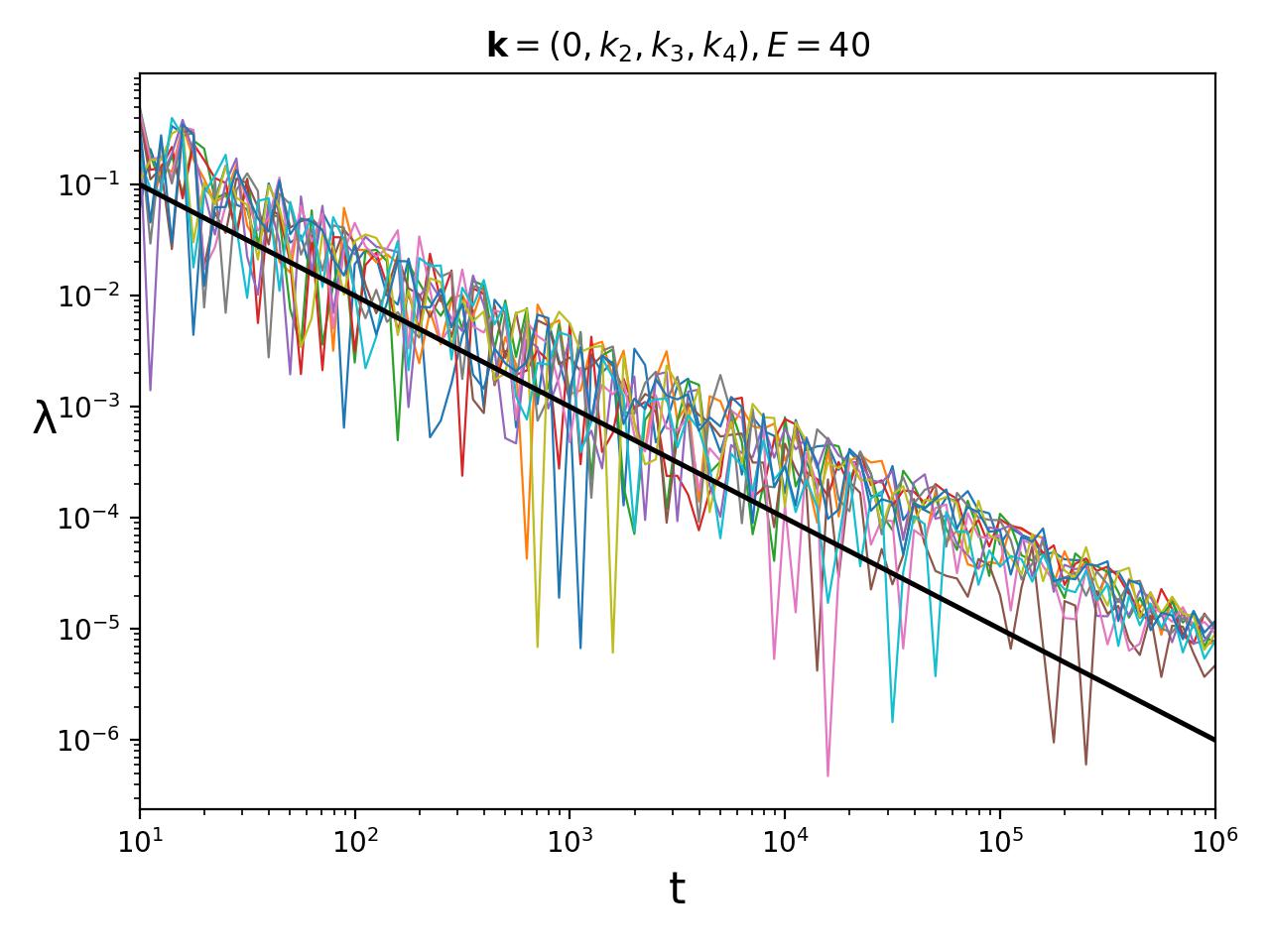}
        \includegraphics[width=0.4\linewidth]{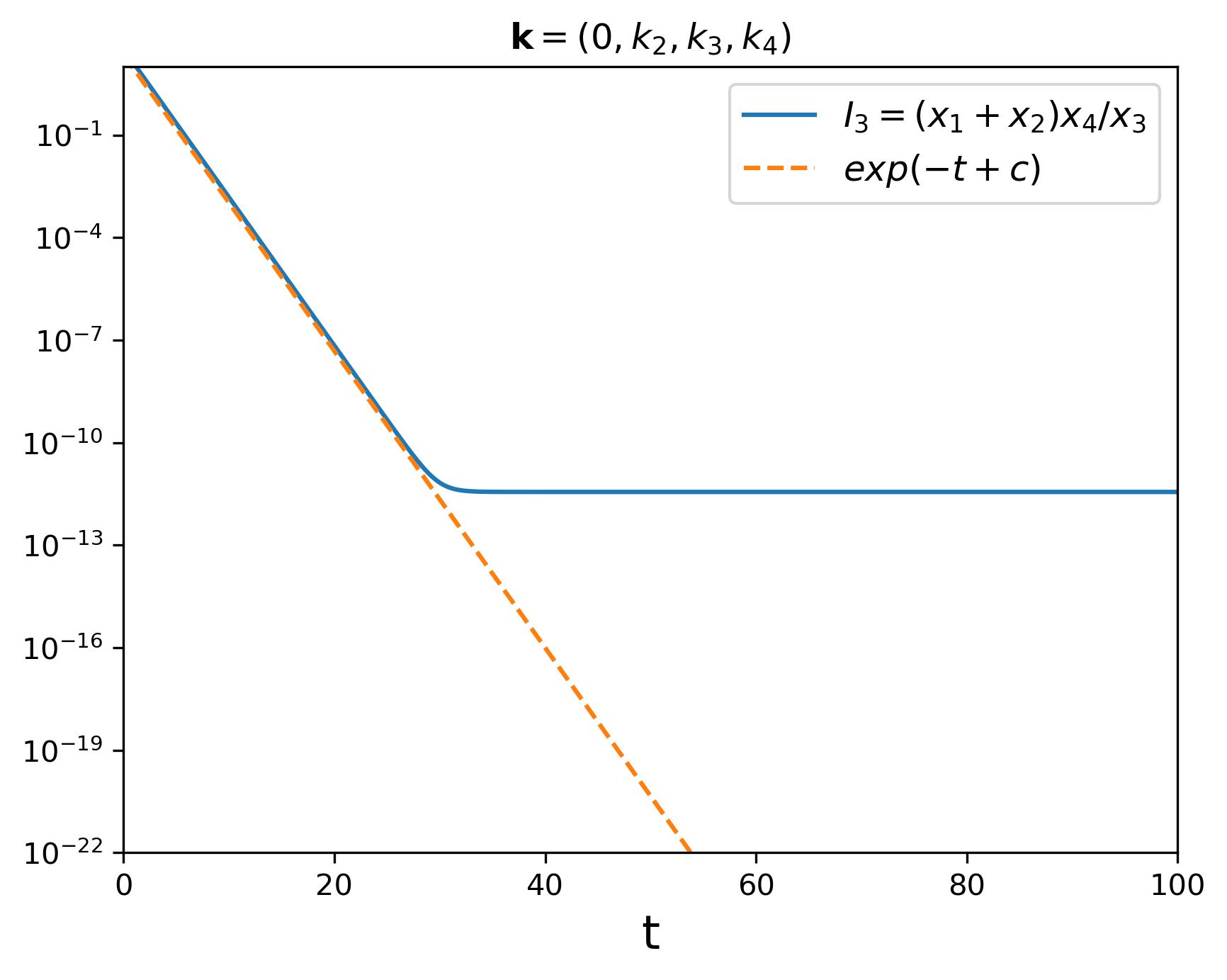}
        \caption{ The system (\ref{oursystem}) with 
        $\mathbf{k}= (0, -1, -\sqrt{2}, -\sqrt{3})$ at energy $E=40$ and random initial conditions. Upper left: Evolution of the variables $\log{ x_i}(t)$. 
        Upper right: The 3D projection of this trajectory in  $(\log{x_1}, \log{x_3}, \log{x_4})$ space. Lower left: 
        Several finite-time maximum Lyapunov exponents for random initial conditions decay as $1/t$. Lower right:
        The evolution of $I_3(t)$ along the dynamics of this three-parametric case decays exponentially with time and converges to a constant around $t=30$. 
        \label{newcase}   }
\end{figure}

All cases not satisfying the two conditions stated in  Proposition 1 for the $\mathbf{k}$ values, i.e. $\alpha, \beta>0$, yield non-sustainable solutions. This is due the stability analysis of equilibrium points, as explained in Section \ref{stabanal}.    
The main open question now concerns the remaining cases, which do not possess an elliptic equilibrium point and have not been identified as any of the known integrable cases discussed in  Section \ref{integrablecases}.
In this Section, we discuss a representative three-parametric case, namely  $\mathbf{k}=(0,k_2,k_3,k_4)$, where the system does not possess an elliptic equilibrium point. Nevertheless, for suitable choices of the parameters $k_i$, some of the populations survive at all times (partially around elliptic points). 

Let us now consider the case with $\mathbf{k}= (0, -1, -\sqrt{2}, -\sqrt{3})$ at high energy, $E=40$. The numerical solution of the system under random initial conditions is displayed in Fig.\ref{newcase}(a), where the evolution of three populations $x_2(t),x_3(t),x_4(t)$ remains bounded by stabilizing to a periodic pattern after times $t>30$, while $x_1(t)$ decays to zero at an exponential rate. The 3D projection of this orbit in $(\ln{x_1}, \ln{x_3}, \ln{x_4})$ space winds around the $\ln{x_1}$--axis in an irregular pattern,  during the first few loops.  It then converges to a periodic structure in $(\ln{x_3}, \ln{x_4})$ plane, while linearly decaying in time with respect to $\ln x_1 (t)$. The initial condition is indicated by the red bullet point, and the arrow of time by the color transition from dark blue to yellow.

Let us examine now the Lyapunov exponents of several random initial conditions at energy $E=40$, to determine whether the system can exhibit chaotic behavior. As Fig.\ref{newcase}(c) shows, the  finite-time Lyapunov functions decay as $1/t$, indicating regular rather than chaotic dynamics. Since three-parametric integrable cases have not been identified, we tested the evolution of the existing integrals of motion listed  in Table \ref{table2} along the orbit of Fig.\ref{newcase}(a) and (b). 
 It is worth noting that most  integrals in Table \ref{table2} oscillate and fluctuate, while some either decay to zero or increase exponentially. There was a single suitable integral, namely the integral $I_3$ associated with the integrable two-parametric case $\mathbf{k}=(0,0,k_3,k_4)$, which  decays as $\exp{(-t+3.513)}$ up to times approximately equal to $t=30$, a time interval that coincides with the `irregular behavior' of the orbit. It then stabilizes to a constant value with very high accuracy (Fig. \ref{newcase}(d)).    

From the above example concerning the $\mathbf{k}=(0,k_2,k_3,k_4)$ three-parametric case, we understand that cases falling outside the realm of Proposition 1 are either known integrable ones, or cases that have not been identified as integrable, yet nevertheless `converge to integrable ones' very fast. In practice, this means that the system with $\mathbf{k}=(0,k_2,k_3,k_4)$, after a short transient time, behaves as if $\mathbf{k}=(0,0,k_3,k_4)$.

\section{Conclusions}
In this paper, we examined the 4-dimensional generalization of Lotka--Volterra models with linear growth rates and constant all-to-all interactions. Our primary focus was on the conditions leading to the long-term survival of all four populations, as well as identifying integrable and non-integrable cases. By investigating the parameter values that give rise to integrable behavior, we found that all integrable cases yield two-parametric families of four-dimensional Lotka--Volterra systems. We then extended our study to three-parametric families, which do not fall into the class of known integrable ones, and found that those without elliptic equilibrium points do not exhibit chaotic behavior. Instead, they converge rapidly to known integrable two-parametric cases.

An interesting conclusion is that the integrable and organized cases do not possess bounded solutions or a balanced coexistence of all four populations. The only sustainable solutions were found in systems with elliptic equilibrium points, which do indeed display chaotic behavior.
Our analysis reveals that ecological resilience and long-term multi-species coexistence are connected with the breakdown of exact integrability. Sustainable, bounded solutions emerge exclusively around the unique elliptic point $x_0=-\mathbf{k}$. In these parameter regimes, the intricate interplay between linear and non-linear components of the Hamiltonian gives rise to chaotic dynamics, unambiguously characterized by positive maximal Lyapunov exponents ($\lambda >0$) and the complex topological structures observed in the Poincar{\'e} surfaces of section. 

This leads to a profound and counterintuitive ecological paradigm: Hamiltonian chaos is not a destabilizing force. Rather, constrained within the energy manifold $H^{-1} (E)$, chaotic behavior provides the essential topological mixing required to steer the system away from the invariant extinction boundaries. Ultimately, while mathematically elegant integrable solutions do describe fragile, doomed ecosystems, Hamiltonian chaos serves as the fundamental dynamical framework for the study of a robust, self-regulating, and sustainably coexisting ecological network.

\section{Acknowledgments}

HC is supported by the EPSRC New Investigator Award UKRI3307 {\it Ergodicity and Lyapunov Exponents in Many-Body Hamiltonian Systems}.

\bibliographystyle{unsrtnat}
\bibliography{references}

\end{document}